\documentclass[a4paper,11pt,reqno]{amsart}
\usepackage[utf8]{inputenc}
\usepackage{mathrsfs}
\usepackage{dsfont}
\usepackage{hyperref}
\usepackage{amsmath}
\usepackage{amssymb}
\usepackage{amsthm}
\usepackage{amsfonts}
\usepackage{amstext}
\usepackage{amsopn}
\usepackage{amsxtra}
\usepackage{mathrsfs}
\usepackage{dsfont}
\usepackage{esint}
\usepackage{enumitem}
\usepackage[capitalize]{cleveref}
\usepackage{braket}

\usepackage{xcolor}% provides \colorlet
\usepackage{fixme}
\fxsetup{
    status=draft,
    author=,
    layout=inline,
    theme=color
}
\definecolor{fxnote}{rgb}{0.8000,0.0000,0.0000}
\colorlet{fxnotebg}{yellow}

\newtheorem{lemma}{Lemma}
\newtheorem{theorem}[lemma]{Theorem}
\newtheorem{proposition}[lemma]{Proposition}
\newtheorem{corollary}[lemma]{Corollary}
\newtheorem{remark}[lemma]{Remark}

\newcommand{\R}{\mathbb{R}}

\newcommand{\C}{\mathbb{C}}

\newcommand{\ii}{\infty}
\newcommand\1{{\ensuremath {\mathds 1} }}
\renewcommand\phi{\varphi}
\newcommand{\gH}{\mathfrak{H}}
\newcommand{\gS}{\mathfrak{S}}

\newcommand{\cE}{\mathcal{E}}

\newcommand{\alp}{\boldsymbol{\alpha}}
\newcommand\pscal[1]{{\ensuremath{\left\langle #1 \right\rangle}}}
\newcommand{\norm}[1]{ \left\| #1 \right\|}
\newcommand{\tr}{{\rm Tr}\,}

\renewcommand{\geq}{\geqslant}
\renewcommand{\leq}{\leqslant}

\newcommand{\eps}{\varepsilon}
\usepackage{color}

\newcommand{\nn}{\nonumber}

\newcommand{\Dk}{D_{\kappa}}

\newcommand{\s}{\mathfrak{S}}

\newcommand{\dd}{\mathrm{d}}

\title{The Scott correction in Dirac-Fock theory}

\author[S. Fournais]{S{\o}ren Fournais}
\address{Department of Mathematics, Aarhus University, Ny Munkegade 118, DK-8000 Aarhus C, Denmark} 
\email{fournais@math.au.dk}

\author[M. Lewin]{Mathieu Lewin}
\address{CNRS \& CEREMADE, Paris-Dauphine University, PSL University, 75016 Paris, France} 
\email{mathieu.lewin@math.cnrs.fr}

\author[A. Triay]{Arnaud Triay}
\address{CEREMADE, Paris-Dauphine University, PSL University, 75016 Paris, France} 
\email{triay@ceremade.dauphine.fr}

\date{\today}

\begin{document}
 
 \begin{abstract}
We give the first derivation of the Scott correction in the large-$Z$ expansion of the energy of an atom in Dirac-Fock theory without projections. 

\bigskip

\noindent \sl \copyright~2019 by the authors. This paper may be reproduced, in its entirety, for non-commercial purposes.
 \end{abstract}

 \maketitle

 \tableofcontents

%%%%%%%%%%%%%%%%%%%%%%%%%%%%%%%%%%%%%%%%%%%%%%%%%%%%%%%%%
%%%%%%%%%%%%%%%%%%%%%%%%%%%%%%%%%%%%%%%%%%%%%%%%%%%%%%%%% 
\section{Introduction} 
%%%%%%%%%%%%%%%%%%%%%%%%%%%%%%%%%%%%%%%%%%%%%%%%%%%%%%%%%
%%%%%%%%%%%%%%%%%%%%%%%%%%%%%%%%%%%%%%%%%%%%%%%%%%%%%%%%%

An impressive success of the many-particle Schr\"odinger equation is its theoretical ability to describe all the atoms of the periodic table. This model has no other parameter than the integer $N=Z$, where $N$ is the number of electrons which is equal to the number of protons $Z$ in a neutral atom. Unfortunately, the exponentially increasing complexity of the problem in $N$ makes any precise computation of the $N$-particle wavefunction impossible in practice. It is therefore important to rely on approximate models and to know whether the true equation simplifies in some limits. 

The most famous result in this direction is the Lieb-Simon proof~\cite{LieSim-73,LieSim-77b,LieSim-77,Lieb-81b} of the leading asymptotics of the ground state energy of a non-relativistic atom with $N$ quantum electrons and a pointwise nucleus of charge $Z=N$,
\begin{equation}
E_{\rm NR}(N,Z=N)=Z^{\frac73}e_{\rm TF}+o(Z^{\frac73})
\label{eq:TF_NR}
\end{equation}
where
\begin{multline*}
e_{\rm TF}=\min_{\substack{\rho\geq0\\ \int_{\R^3}\rho=1}}\bigg\{\frac3{10}(3\pi^2)^{\frac23} \int_{\R^3}\rho(x)^{\frac53}\,dx-\int_{\R^3}\frac{\rho(x)}{|x|}\,dx\\
+\frac12\iint_{\R^3\times\R^3}\frac{\rho(x)\rho(y)}{|x-y|}\,dx\,dy\bigg\}
\label{eq:TF_energy}
\end{multline*}
is the minimum Thomas-Fermi energy~\cite{Thomas-27,Fermi-27}. Thomas-Fermi theory does not only provide the leading order of the energy. It also describes the precise behavior of the density of electrons at the distance $Z^{-1/3}$ to the nucleus and it is believed to provide a surprisingly accurate estimate on the size of alkali atoms~\cite{Solovej-16}. 

The expansion~\eqref{eq:TF_NR} has been continued in many works. The best result known at the moment is
\begin{equation}
E_{\rm NR}(N,Z=N)=Z^{\frac73}e_{\rm TF}+\frac{Z^2}{2}+Z^{\frac53}c_{\rm DS}+o(Z^{\frac53}).
\label{eq:TFSDS_NR}
\end{equation}
The $Z^2$ term is the \emph{Scott correction}~\cite{Scott-52,Schwinger-80} which is the main subject of this article and was rigorously derived in~\cite{SieWei-87,SieWei-87b,SieWei-89,Hugues-90,IvrSig-93,IanLieSie-95}. This was then generalized in several directions~\cite{Bach-89,SolSpi-03}, including for magnetic fields~\cite{Ivrii-96,Ivrii-97,Sobolev-96}. The next order $Z^{5/3}$ contains both an exchange term predicted by Dirac~\cite{Dirac-28b} and a semi-classical correction derived by Schwinger~\cite{Schwinger-81,EngSch-84a,EngSch-84b,EngSch-84c}. It was rigorously established in an impressive series of works by Fefferman and Seco~\cite{FefSec-89,FefSec-90,FefSec-92,FefSec-93,FefSec-94,FefSec-94b,FefSec-94c,FefSec-94d}. It should be mentioned that although the leading $Z^{7/3}$ Thomas-Fermi term and the $Z^{5/3}$ Dirac exchange term are somewhat universal (that is, arise for other types of interactions in mean-field limits~\cite{FouLewSol-18,Bach-92,GraSol-94,BenNamPorSchSei-18}), the $Z^2$ Scott correction and the $Z^{5/3}$ Schwinger term are specific to the Coulomb potential. More precisely, these are semi-classical corrections due to the singularity of the Coulomb potential at the origin. It should also be noted that the three leading terms in~\eqref{eq:TF_NR} are already correctly described by Hartree-Fock theory~\cite{Bach-92,Bach-93,GraSol-94}. The exchange term only participates to the $Z^{5/3}$ term and it can be dropped for the first two terms~\cite{LieSim-77}, leading to the so-called \emph{reduced} Hartree-Fock model~\cite{Solovej-91}. 

It is well known in Physics and Chemistry that, in heavy atoms, relativistic effects start to play an important role, even for not so large values of $Z$. Without relativity, gold would have the same color as silver~\cite{GlaAmb-10}, mercury would not be liquid at room temperature~\cite{CalPahWorSch-13} and cars would not start~\cite{ZalPyy-11}. The reason why relativistic effects become important is because, in an atom, most of the electrons live at a distance $Z^{-1/3}$ to the nucleus, hence they experiment very strong Coulomb forces leading to very high velocities, of the order of the speed of light. This is even more dramatic for the Scott correction which is due to the few electrons living at the very short distance $Z^{-1}$ to the nucleus. Indeed, Schwinger has predicted in~\cite{Schwinger-80} that small relativistic effects should not affect the leading Thomas-Fermi energy in the large-$Z$ expansion, but should modify the Scott correction. 

A truly relativistic model should involve the Dirac operator~\cite{Thaller,EstLewSer-08}. Unfortunately, there is no well-defined $N$-particle Dirac Hamiltonian at the moment~\cite{Derezinski-12}, except for $N=2$~\cite{DecOel-19}, and even if there was one it would probably have no bound state. In the very unlikely case of the existence of bound states, it would anyway be impossible to identify a ground state. The one-particle Dirac operator is unbounded both from above and below and any $N$-particle Dirac Hamiltonian would have the whole line as its spectrum. A better theory should involve bound states in Quantum Electrodynamics~\cite{Shabaev-02}, but this is far from being understood mathematically. 

Several authors have instead studied the expansion of the ground state for simplified relativistic models. S\o{}rensen studied a pseudo-relativistic Hamiltonian where the Laplacian is replaced by a non-local fractional Laplacian (the ``Chandrasekhar'' operator), and proved that the leading Thomas-Fermi term is unchanged in this case~\cite{Sorensen-05}. The Scott correction for this model was then derived in~\cite{Sorensen-98,SolSorSpi-10,FraSieWar-08}, but it does not coincide with Schwinger's original prediction~\cite{Schwinger-80}, since the spectral properties of the Dirac operator and of the fractional Laplacian are different. Siedentop and co-workers~\cite{CasSie-06,FraSieWar-09,HanSie-15} have then considered \emph{projected Dirac operators} in order to suppress its negative spectrum, in the spirit of Brown-Ravenhall~\cite{BroRav-51} and Mittleman~\cite{Mittleman-81}. However, the Scott correction depends in a non trivial way of the chosen projection, which is somewhat arbitrary. The expected relativistic Scott correction has been obtained in the recent work~\cite{HanSie-15} which covers the larger class of projections and in particular includes the positive spectral projection of the non-interacting Dirac-Coulomb operator, which happens to give the correct Scott term. However, discrepancies could re-appear in the next order term for this projection.

In this paper we provide the first rigorous derivation of the relativistic Scott correction in (reduced) Dirac-Fock theory without projection. As we have said, we cannot start with the ill-defined $N$-body Dirac theory. However, let us recall that the Scott correction is already fully included in mean-field theory, even without exchange term. In all the previous works on the Scott correction, the reduction from the $N$-particle Schr\"odinger Hamiltonian to the (reduced) Hartree-Fock ground state is usually an easy step. In the non-relativistic case, it for instance immediately follows from the Lieb-Oxford inequality~\cite{Lieb-79,LieOxf-80}. For this reason, it makes sense to directly start with (reduced) Dirac-Fock theory and prove the Scott correction within this theory. In order to simplify our exposition we discard the exchange term completely but we expect the same results when it is included. The exchange term is a lower order correction. Note that the predictions of Dirac-Fock theory for the Scott correction agree quite well with experimental data for $Z=1,...,120$, according to~\cite{Desclaux-73} and as was discussed in~\cite[Sec.~6]{HanSie-15}.

Dirac-Fock theory is the relativistic counterpart of the Hartree-Fock model and it has the advantage of having well-defined solutions which can be interpreted as ground states, even though the corresponding energy functional is unbounded from below~\cite{EstSer-99,Paturel-00,EstSer-01,EstSer-02a,HubSie-07,EstLewSer-08,Sere-09}. Those correspond to electronic states in the positive spectral subspace of their own mean-field Dirac operator. Hence this theory does rely on a projection but it is unknown \emph{a priori} and depends in a nonlinear way on the solution itself. Our task will therefore be to extend the result of Handrek and Siedentop~\cite{HanSie-15} for a fixed projection to the case of a self-consistent projection depending on the density matrix of the system. We will estimate the energy cost of replacing the self-consistent projection by the fixed Dirac-Coulomb projection in order to apply the result from~\cite{HanSie-15}. 

The paper is organized as follows. In the next section we first state some spectral properties of Dirac operators with Coulomb potentials before we are able to properly introduce the Dirac-Fock minimization principle and finally give its large-$Z$ expansion. The rest of the paper is devoted to the proof of our main result on the Scott correction.

\bigskip

\noindent\textbf{Acknowledgments.} We thank \'Eric S\'er\'e for fruitful discussions and for providing us with his unpublished work~\cite{Sere-09}. This project has received funding from the European Research Council (ERC) under the European Union's Horizon 2020 Research and Innovation Programme (Grant agreement MDFT No 725528).

%%%%%%%%%%%%%%%%%%%%%%%%%%%%%%%%%%%%%%%%%%%%%%%%%%%%%%%%%
%%%%%%%%%%%%%%%%%%%%%%%%%%%%%%%%%%%%%%%%%%%%%%%%%%%%%%%%% 
\section{Main result} 
%%%%%%%%%%%%%%%%%%%%%%%%%%%%%%%%%%%%%%%%%%%%%%%%%%%%%%%%%
%%%%%%%%%%%%%%%%%%%%%%%%%%%%%%%%%%%%%%%%%%%%%%%%%%%%%%%%%

%%%%%%%%%%%%%%%%%%%%%%%%%%%%%%%%%%%%%%%%%%%%%%%%%%%%%%%%%
\subsection{Gaps in Dirac-Coulomb operators}
%%%%%%%%%%%%%%%%%%%%%%%%%%%%%%%%%%%%%%%%%%%%%%%%%%%%%%%%%

In this section we discuss some important properties of Dirac operators with Coulomb potentials that will be important in our situation. Several tools introduced here are taken from~\cite{EstLewSer-19b_ppt}.

Throughout the whole section we will be looking at operators in the form
\begin{equation}
D_0+\rho\ast\frac{1}{|x|} 
\label{eq:Dirac_two_charges}
\end{equation}
where $\rho$ is a signed bounded measure in $\R^3$. We are typically interested in the case where $\rho=\alpha\rho_+-\kappa\delta_0$, with $\kappa\delta_0$ the density of the point nucleus and $\rho_+\in (L^1\cap L^{3/2})(\R^3)$ a more regular measure describing the quantum electrons. Here and everywhere, we work in a system of units such that $\hbar=m=c=1$. Then we have $\kappa=\alpha Z$ where $Z$ is the number of protons and $\alpha=e^2\simeq 1/137.04$ is the Sommerfeld fine structure constant, which is the square of the charge of the electron. We recall that the free Dirac operator $D_0$ in 3d is given by
\begin{equation}
D_0\ = -i\; \boldsymbol{\alpha}\cdot\boldsymbol{\nabla} + \beta = \ - i\
\sum^3_{k=1} {\bf \alpha}_k \partial _k + {\bf \beta},
\label{def_Dirac}
\end{equation}
where $\alpha_1$, $\alpha_2$, $\alpha_3$ and $\beta$ are $4\times4$ Hermitian matrices satisfying the anticommutation relations
\begin{equation} \label{CAR}
\left\lbrace
\begin{array}{rcl}
 {\alpha}_k
{\alpha}_\ell + {\alpha}_\ell
{\alpha}_k  & = &  2\,\delta_{k\ell}\,\1_{\C^4},\\
 {\alpha}_k {\beta} + {\beta} {\alpha}_k
& = & 0,\\
\beta^2 & = & \1_{\C^4}.
\end{array} \right. \end{equation}
The usual representation in $2\times 2$ blocks is given by 
$$ \beta=\left( \begin{matrix} I_2 & 0 \\ 0 & -I_2 \\ \end{matrix} \right),\quad \; \alpha_k=\left( \begin{matrix}
0 &\sigma_k \\ \sigma_k &0 \\ \end{matrix}\right),  \qquad k=1, 2, 3\,,
$$
with the Pauli matrices
$$\sigma _1=\left( \begin{matrix} 0 & 1
\\ 1 & 0 \\ \end{matrix} \right),\quad  \sigma_2=\left( \begin{matrix} 0 & -i \\
i & 0 \\  \end{matrix}\right),\quad  \sigma_3=\left( 
\begin{matrix} 1 & 0\\  0 &-1\\  \end{matrix}\right) \, .$$
The operator $D_0$ is self-adjoint in $L^2(\R^3,\C^4)$ with domain $H^1(\R^3,\C^4)$ and its spectrum is $\sigma(D_0)=(-\ii,-1]\cup[1,\ii)$, see~\cite{Thaller,EstLewSer-08}. 

If $|\rho|^{1/2}\in H^{1/2}(\R^3)$, then by the Hardy-Kato inequality
\begin{equation}
\frac{1}{|x|}\leq\frac\pi2 \sqrt{-\Delta},
\label{eq:Kato}
\end{equation}
the Coulomb potential is in $L^\ii(\R^3)$, with the pointwise bound
\begin{equation}
\left|\rho\ast\frac1{|x|}\right|\leq \frac\pi2 \pscal{\sqrt{|\rho|},\sqrt{-\Delta}\sqrt{|\rho|}}.
\label{eq:pointwise_bd_V}
\end{equation}
Hence $D_0+\rho\ast|x|^{-1}$ is self-adjoint on the same domain $H^1(\R^3,\C^4)$. When 
$$\rho=\rho_+-\kappa\delta_0$$
with $\rho_+\geq0$, $\rho_+^{1/2}\in H^{1/2}(\R^3)$ and $0\leq\kappa<1$, then it immediately follows that $D_0+\rho_+\ast|x|^{-1}-\kappa|x|^{-1}$ is self-adjoint on the same domain as $D_0-\kappa|x|^{-1}$. The latter operator has a unique distinguished self-adjoint extension on $H^1(\R^3,\C^4)$, whose domain is always included in $H^{1/2}(\R^3,\C^4)$. We refer for instance to~\cite[Section~1]{EstLewSer-19} for a review of important properties of such operators. 

The condition that $\kappa=\alpha Z<1$ means that in principle we cannot consider atoms with nuclear charge higher than $137$. In order to relate the Dirac-Coulomb model to its non-relativistic counterpart, we will take $\alpha\to0$ (non-relativistic limit) at the same time as $Z\to\ii$, while keeping $\kappa=\alpha Z$ fixed. This is the natural limit for the large-$Z$ expansion of relativistic systems.

The following well-known result is a more quantitative expression of the fact that the domain of $D_0+\rho\ast|x|^{-1}$ contains $H^1(\R^3,\C^4)$. 

\begin{lemma}[Upper bound on $(D_{\rho})^2$]
Let $\rho$ be a signed, bounded measure on $\R^3$. For every $\Psi\in H^1(\R^3,\C^4)$, we have 
\begin{equation}
\norm{\left(D_0+\rho\ast\frac{1}{|x|} \right)\Psi}_{L^2(\R^3,\C^4)}\leq \left(1+2|\rho|(\R^3)\right)\norm{D_0\Psi}_{L^2(\R^3,\C^4)}.
\label{eq:upper_bd_D_Psi}
\end{equation}
Hence, in the situations recalled above where $D_0+\rho\ast|x|^{-1}$ has a distinguished self-adjoint extension on $H^1(\R^3,\C^4)$, we have the operator inequality
\begin{equation}
\left(D_0+\rho\ast\frac{1}{|x|} \right)^2\leq \left(1+2|\rho|(\R^3)\right)^2|D_0|^2
\label{eq:upper_bd_D2}
\end{equation}
and
\begin{equation}
\left|D_0+\rho\ast\frac{1}{|x|} \right|\leq \left(1+2|\rho|(\R^3)\right)|D_0|.
\label{eq:upper_bd_D}
\end{equation}
\end{lemma}

\begin{proof}
The estimate (\ref{eq:upper_bd_D_Psi}) follows from Hardy's inequality $|x|^{-2}\leq 4(-\Delta)\leq 4|D_0|^2$. The last inequality~\eqref{eq:upper_bd_D} is a consequence of~\eqref{eq:upper_bd_D2} since the square root is operator monotone.
\end{proof}

The purpose of the section is to discuss lower bounds similar to~\eqref{eq:upper_bd_D}. From now on, we use the shorthand notation
$$\boxed{D_\kappa:=D_0-\frac{\kappa}{|x|}}$$
for the usual Dirac-Coulomb operator, with $0\leq\kappa<1$, and
$$\boxed{D_{\kappa,\rho}:=D_0-\frac{\kappa}{|x|}+\rho\ast\frac{1}{|x|}}$$
when it is perturbed by a density $\rho$ (typically positive and regular enough in our context). 

We recall that the lowest eigenvalue of $D_\kappa$ in the gap $[-1;1]$ is $\sqrt{1-\kappa^2}$ and that the operator has an increasing sequence of eigenvalues tending to the upper threshold $1$~\cite{Thaller}. In addition, it was proved in~\cite{MormUl-17} that for all $0\leq\kappa<1$, there exists a constant $c_\kappa>0$ so that 
\begin{equation}
c_\kappa|D_0|\leq |D_\kappa|.
\label{eq:compare_Dirac_kappa}
\end{equation}
When $\kappa\geq\sqrt{3}/2$, $(D_\kappa)^2$ cannot be lower bounded by $(D_0)^2$, otherwise the domain would be equal to $H^1(\R^3,\C^4)$. However, due to the explicit form of the domain of $D_\kappa$ as explained in~\cite{EstLewSer-19}, we indeed have
\begin{equation}
 c_{\kappa}(s)|D_0|^{2s}\leq |D_\kappa|^2
 \label{eq:lower_bound_Dirac_Coulomb_powers}
\end{equation}
for all $0\leq s<\min(1,1/2+\sqrt{1-\kappa^2})$. By interpolation, this gives
\begin{equation}
 c'_{\kappa}(s)|D_0|^{1+\eta(2s-1)}\leq |D_\kappa|^{1+\eta}, \quad \forall \eta \in [0,1].
 \label{eq:lower_bound_Dirac_Coulomb_powers_eta}
\end{equation}

A natural question, which will play an important role later in our study, is to ask how big are the eigenvalues of $D_{\kappa,\rho}$ for a \emph{general} positive density $\rho$. For which $\rho$ can one guarantee that the gap around the origin is preserved? Following~\cite{EstLewSer-19b_ppt} we introduce a critical value $\nu_0$ that works for $\kappa\equiv0$, before looking at the case $\kappa>0$. 

Let $\rho$ be a non-negative density such that $\sqrt\rho\in H^{1/2}(\R^3)$ and $\int_{\R^3}\rho=1$. Since the associated Coulomb potential is bounded uniformly by~\eqref{eq:pointwise_bd_V}, the eigenvalues of $D_{0,\nu\rho}=D_0+\nu\rho\ast|x|^{-1}$ are all confined to an interval of size proportional to $\nu$ at the edges of the gap $[-1,1]$, for $\nu$ small enough. Actually, the min-max characterization of the eigenvalues from~\cite{DolEstSer-00} implies that there is no eigenvalue close to $1$ and that there are infinitely many close to $-1$, since the potential is repulsive. In addition, these eigenvalues are monotonically increasing with $\nu$. Let then $\nu_0(\rho)$ be the first value of the coupling constant $\nu$, for which the largest negative eigenvalue vanishes:
$$\nu_0(\rho):=\min\big\{\nu>0\ :\ 0\in\sigma(D_{0,\nu\rho})\big\}.$$
Let finally
\begin{equation}
\boxed{\nu_0:=\inf_{\substack{\rho\geq0\\ \rho(\R^3)=1 \\ \sqrt{\rho} \in H^{1/2}(\mathbb{R}^3)}}\nu_0(\rho)}
\label{eq:critical_nu}
\end{equation}
be the lowest possible critical value among all probability densities. Loosely speaking, $\nu_0$ is the largest possible repulsive charge that we can add while guaranteeing that the eigenvalues will stay in $[-1,0]$, independently of the shape of the density $\rho$. By charge conjugation, we get the reverse picture if we place an arbitrary attractive charge, that is, we allow negative $\nu$'s. 

All our next results will be stated in terms of this critical $\nu_0$. Following~\cite{EstLewSer-19}, we conjecture that $\nu_0=1$, that is, the worse case is when the density is a Dirac delta. The following is shown in~\cite{EstLewSer-19b_ppt}.

\begin{lemma}[Estimates on $\nu_0$~\cite{EstLewSer-19b_ppt}]
We have
\begin{equation}
0.91\simeq \frac{2}{\frac\pi2+\frac2\pi}\leq \nu_0\leq1.
\label{eq:estim_nu_c}
\end{equation}
\end{lemma}

\begin{proof}
The upper bound is obtained by concentrating $\rho$ at the origin to make it converge to $\delta_0$ and by using the exact Coulomb value. The lower bound follows from Tix's inequality~\cite{Tix-98}
\begin{equation}
P_0^+\left(D_0-\frac\kappa{|x|}\right)P_0^+\geq(1-\kappa)P^+_0,\qquad\forall 0\leq\kappa\leq \frac{2}{\frac\pi2+\frac2\pi},
\label{eq:Tix}
\end{equation}
where $P_0^+=\1(D_0\geq0)$ is the positive spectral projection of the free Dirac operator. By translation invariance we deduce that 
\begin{equation}
P_0^+\left(D_0-\nu\rho\ast\frac1{|x|}\right)P_0^+\geq\left(1-\nu\right)P^+_0
\label{eq:Tix2}
\end{equation}
for every probability measure $\rho$.
Let us now explain how to derive the lower bound in~\eqref{eq:estim_nu_c} using~\eqref{eq:Tix2}.
We can use the min-max characterization from~\cite{DolEstSer-00}, as described also in~\cite{EstLewSer-19}. For simplicity we work with $D_{0,-\nu\rho}=D_0-\nu\rho\ast|x|^{-1}$ instead of $D_{0,\nu\rho}=D_0+\nu\rho\ast|x|^{-1}$, which is the same by charge conjugation. We define 
$$a_-:=\max_{\substack{\Psi\in P^-_0L^2\\ \|\Psi\|=1 }}\pscal{\Psi,\left(D_0-\nu\rho\ast|x|^{-1}\right)\Psi}$$
which satisfies $a_-=-1$ due to the negative sign of the potential $-\rho\ast|x|^{-1}$. Then we look at the min-max
$$\lambda_1(\nu):=\inf_{\Psi_+\in P^+_0L^2}\sup_{\Psi_-\in P^-_0L^2}\frac{\pscal{\Psi_++\Psi_-,\left(D_0-\nu\rho\ast|x|^{-1}\right)(\Psi_++\Psi_-)}}{\|\Psi_++\Psi_-\|^2}.$$
Taking $\Psi_-\equiv0$ and using Tix's inequality~\eqref{eq:Tix2}, we obtain
$$\lambda_1(\nu)\geq \left(1-\nu\right).$$
In particular $\lambda_1(\nu)>a_-$, hence $\lambda_1(\nu)$ is the lowest eigenvalue of $D_{-\nu\rho}$ above $-1$, by~\cite{DolEstSer-00}. This proves as we wanted that $\lambda_1(\nu)>0$ for $\nu<2/(\pi/2+2/\pi)$. In other words, we have $\nu_0(\rho)\geq2/(\pi/2+2/\pi)$. A slightly different proof is given in~\cite{EstLewSer-19b_ppt}.
\end{proof}

We now investigate the case of a negative point charge together with a smoother positive charge, $\rho_+-\kappa\delta_0$. Similarly as before, we introduce
$$\nu_0(\kappa,\rho_+):=\min\big\{\nu>0\ :\ 0\in\sigma(D_{\kappa,\nu\rho_+})\big\}$$
and 
$$\nu_0(\kappa):=\inf_{\substack{\rho_+\geq0\\ \rho_+(\R^3)=1}}\nu_0(\kappa,\rho_+).$$
This critical charge is actually equal to the one at $\kappa=0$.

\begin{lemma}
For every $\kappa\in[0,1)$, we have 
$$\nu_0(\kappa)=\nu_0.$$
\end{lemma}

\begin{proof}
By the min-max principle of~\cite{DolEstSer-00} one can see that the negative eigenvalues of $D_{\kappa,\nu\rho_+}$ are decreasing with respect to $\kappa$ and increasing with respect to $\nu$ (a different proof of this will be given in the proof of Theorem~\ref{thm:Dirac_gap} in Section~\ref{sec:proof_Dirac_gap} below). In particular, we have $\nu_0(\kappa,\rho_+)\geq \nu_0(0,\rho_+)=\nu_0(\rho_+)$. After minimizing over $\rho_+$, this gives $\nu_0(\kappa)\geq\nu_0$. However, by placing $\rho_+$ very far away from the origin, we also see that  $\nu_0(\kappa)\leq\nu_0$ and there must therefore be equality. 
\end{proof}

The main result of this section is that the gap implies a universal operator bound.

\begin{theorem}[Gap of general sub-critical Dirac-Coulomb operators]\label{thm:Dirac_gap}
For every $0\leq\kappa<1$ and every $0\leq \nu<\nu_0$, there exists a universal constant $c_{\kappa,\nu}>0$ so that 
\begin{equation}
c_{\kappa,\nu}|D_0|\leq |D_{\kappa,\rho}|
\label{eq:compare_Dirac}
\end{equation}
for every non-negative $\rho$ so that $\sqrt\rho\in H^{1/2}(\R^3)$ and $\int_{\R^3}\rho=\nu$. One can for instance take 
\begin{equation}
c_{\kappa,\nu}= \frac{c_\kappa^2}{1+ 2(\nu+\kappa)}\bigg(1+\frac{\pi}{2c_\kappa}\frac{\nu\,\nu_0}{\nu_0-\nu}\bigg)^{-2}
\label{eq:value_c_kappa_nu}
\end{equation}
where $c_\kappa$ is the best constant in~\eqref{eq:compare_Dirac_kappa}.
\end{theorem}

The proof is provided in Section~\ref{sec:proof_Dirac_gap}. Theorem~\ref{thm:Dirac_gap} gives a lower bound, similar to the upper bound~\eqref{eq:upper_bd_D}, which is completely independent of the shape of the negative charge $\rho$ and this will play a decisive role in the study of our nonlinear problem, where $\rho$ is unknown. The constant in~\eqref{eq:value_c_kappa_nu} is not at all optimal and it is only displayed for concreteness. 

\begin{remark}[More general negative densities]
The theorem applies to all positive Borel measures $\rho$. It is not at all necessary that $\sqrt\rho\in H^{1/2}(\R^3)$. But we have stated it in this context since this is what we will need later, and because working with a bounded repulsive Coulomb potential simplifies some technical arguments. 
\end{remark}

\begin{remark}[A better estimate on the gap of $D_{\kappa,\rho}$]\label{rmk:better_estimate_gap}
The estimate~\eqref{eq:compare_Dirac} implies the following estimate on the gap around the origin:
$$\sigma(D_{\kappa,\rho})\cap\big(-c_{\kappa,\nu};c_{\kappa,\nu}\big)=\emptyset.$$
However, we indeed have 
$$\sigma(D_{\kappa,\rho})\cap\big(\lambda_c(\nu);\sqrt{1-\kappa^2}\big)=\emptyset,$$
where 
$$\lambda_c(\nu):=\sup_{\substack{\rho\geq0\\ \int_{\R^3}\rho=\nu}}\lambda_1(D_{0,\rho})<0.$$
is the largest possible value of the last negative eigenvalue, when optimized over all densities $\rho$ with subcritical mass $\int_{\R^3}\rho=\nu<\nu_0$. 
This is due to the min-max principle of~\cite{DolEstSer-00}, which implies that the negative eigenvalues are decreasing in $\kappa$ at fixed $\rho$ and that the positive eigenvalues are increasing in $\nu$ at fixed $\kappa$, as we have already mentioned. 
\end{remark}

\begin{remark}[More general positive densities]
We have considered an attractive Dirac charge because this is what will be needed later. It is also possible to show that for any $\rho_\pm\geq0$ with $\int_{\R^3}\rho_\pm=\nu_\pm<\nu_0$, we have  with $\rho=\rho_+-\rho_-$
$$c_{\nu_+,\nu_-}|D_0|\leq\left|D_0+\rho\ast\frac1{|x|}\right|$$ 
for some $c_{\nu_+,\nu_-}>0$ and 
$$\sigma\left(D_0+\rho\ast\frac1{|x|}\right)\cap\big(-\lambda_c(\nu_+);\lambda_c(\nu_-)\big)=\emptyset.$$
See~\cite{EstLewSer-19b_ppt}.
\end{remark}

%%%%%%%%%%%%%%%%%%%%%%%%%%%%%%%%%%%%%%%%%%%%%%%%%%%%%%%%%
\subsection{A minimization principle in Dirac-Fock theory}
%%%%%%%%%%%%%%%%%%%%%%%%%%%%%%%%%%%%%%%%%%%%%%%%%%%%%%%%%

The Dirac-Fock energy functional is obtained from the Hartree-Fock model by replacing $-\Delta/2$ by $D_0-1$. In the spirit of~\cite{Lieb-81,BacLieSol-94}, we consider mixed quasi-free states and express everything in terms of the one particle density matrix $\gamma$, which is a bounded self-adjoint operator on $L^2(\R^3,\C^4)$ such that $0\leq\gamma\leq1$ and $\tr\gamma=N$ for $N$ electrons. The (reduced) Dirac-Fock energy reads 
\begin{equation}
\boxed{\cE^{\rm DF}_{\kappa,\alpha}(\gamma)=\tr(D_\kappa-1)\gamma+\frac{\alpha}{2}\iint_{\R^3\times\R^3}\frac{\rho_\gamma(x)\rho_\gamma(y)}{|x-y|}\,dx\,dy}
\label{eq:DF_energy}
\end{equation}
where 
$$\rho_\gamma(x)=\tr_{\C^4}\gamma(x,x)$$
is the associated density. The functional $\cE^{\rm DF}_{\kappa,\alpha}$ is unbounded from below, due to the fact that $D_\kappa$ is itself unbounded. 
Critical points satisfy the self-consistent equation
\begin{equation*}
[\gamma ,D_{\kappa,\alpha\rho_\gamma}]=0
\end{equation*}
where $D_{\kappa,\alpha\rho_\gamma}=D_\kappa+\alpha\rho_\gamma\ast|x|^{-1}$ is called the \emph{mean-field} or \emph{Fock operator}. 
We are interested in the stationary states satisfying the more precise equation
\begin{equation}
\boxed{\gamma =\1(0\leq D_{\kappa,\alpha\rho_\gamma}\leq\mu),}
\label{eq:DF_SCF}
\end{equation}
that is, $\gamma$ is the orthogonal projection corresponding to the $N$ first \emph{positive} eigenvalues of its Fock operator. Re-expressed in terms of the $N$ first eigenfunctions, this gives a system of $N$ coupled nonlinear Dirac equations
\begin{equation}
\left(D_\kappa+\alpha\sum_{k=1}^N|\phi_k|^2\ast\frac{1}{|x|}\right)\phi_j=\mu_j\phi_j,
\label{eq:DF_SCF_orbitals}
\end{equation}
with $0<\mu_1\leq\cdots \leq\mu_N=\mu$ the $N$ first positive eigenvalues. 
States satisfying~\eqref{eq:DF_SCF} can be interpreted as \emph{ground states} since they solve the same kind of equation as for Hartree-Fock minimizers. It will be useful to introduce the projection
$$P^+_{\kappa,\gamma}:=\1(D_{\kappa,\alpha\rho_\gamma}\geq0)$$
on the positive spectral subspace of the mean-field operator and to note that $\gamma P^+_{\kappa,\gamma}=\gamma$, or equivalently $0\leq \gamma \leq P^+_{\kappa,\gamma}$. 

The existence of infinitely many solutions to~\eqref{eq:DF_SCF_orbitals} was proved for the first time by Esteban-S\'er\'e in~\cite{EstSer-99} (with exchange term, but the argument works the same without exchange) for all 
$$\max(\kappa,3\alpha N)<\frac{2}{\pi/2+2/\pi},\qquad N\leq Z=\kappa/\alpha,$$
using a concavity argument in the negative directions. Although in this work the $\mu_j$ are known to be positive, they are however not necessarily the $N$ first eigenvalues. The result of~\cite{EstSer-99} was then generalized to the range
$$\max(\kappa,\alpha N)<\frac{2}{\pi/2+2/\pi},\qquad N\leq Z=\kappa/\alpha,$$
by Paturel in~\cite{Paturel-00}. That is, the unphysical factor 3 was removed, using a Lyapounov-Schmidt reduction and a linking argument. 
In~\cite{EstSer-01}, Esteban and S\'er\'e have shown that their first solution indeed converges to the (non-relativistic) Hartree-Fock minimizer in the limit $\alpha\to0$ at fixed $Z$ and $N$, after a proper rescaling. Furthermore, the $\mu_j$ are the first eigenvalues of $D_{\kappa,\alpha\rho_\gamma}$ for $\alpha$ small enough. This is the justification that~\eqref{eq:DF_SCF} is the natural equation for a Dirac-Fock ground state. Finally, they proved in the same article~\cite{EstSer-01} that for $\alpha$ small enough, their solution solves the following minimization problem
\begin{equation}
\boxed{E^{\rm DF}(\kappa,\alpha,N):=\min_{\substack{0\leq\gamma\leq P^+_{\kappa,\gamma}\\ \tr(\gamma)\leq N}}\cE_{\kappa,\alpha}^{\rm DF}(\gamma).}
\label{eq:DF_min_energy}
\end{equation}
In words, the Dirac-Fock ground state minimizes the Dirac-Fock energy among all the states which live in the positive spectral subspace of their own mean-field operator $P^+_{\kappa,\gamma}$. This is a very nonlinear constraint but it is physically meaningful. Using a simpler fixed point technique, Huber and Siedentop have later obtained a similar result in~\cite{HubSie-07} for $\alpha$ small enough at fixed $\kappa=\alpha Z$ and $N$, that is, for large atoms but small interactions. Finally, in the unpublished work~\cite{Sere-09}, S\'er\'e has directly studied the minimization problem~\eqref{eq:DF_min_energy} for $\kappa$ and $\alpha N$ fixed but small enough, with quantitative estimates. His argument is based on the function 
$$\theta(\gamma)=\lim_{n\to\ii}\gamma_n$$
where the sequence $\gamma_n$ is recursively defined by
$$\begin{cases}
   \gamma_{n+1}=P_{\kappa,\gamma_{n}}^+\gamma_{n}P_{\kappa,\gamma_{n}}^+\\
   \gamma_0=\gamma.
  \end{cases}$$ 
The function $\theta$ is used to project any $\gamma$ sufficiently close to satisfying the constraint $P^-_{\kappa,\gamma}\gamma=0$ to a new state $\theta(\gamma)$ which does satisfy this constraint. This allows to show that the set of states $\{\gamma\ :\ P^-_{\kappa,\gamma}\gamma=0\}$ is a smooth manifold on which one can use variational techniques. 

To summarize the situation, there are solutions to the Dirac-Fock equations~\eqref{eq:DF_SCF} for $\max(\alpha Z,\alpha N)<2/(\pi/2+2/\pi)$. One of these solutions is known to solve the minimization problem~\eqref{eq:DF_min_energy}, but only in a limiting regime of small $\alpha$ (either with $N$ and $Z$ fixed, or $\kappa$ and $N$ fixed or $\kappa$ and $\alpha N$ small but fixed). In this paper we will study the limit of~\eqref{eq:DF_min_energy} in the whole range 
$$0\leq \kappa<1,\qquad 0\leq \alpha N<\nu_0$$
where $\nu_0$ is the critical number defined in the previous section, although minimizers are not necessarily known to exist in all cases. Note that we have required $\tr(\gamma)\leq N$ instead of $\tr\gamma=N$ in~\eqref{eq:DF_energy}. We believe that minimizers always exist in~\eqref{eq:DF_min_energy} and that they satisfy $\tr\gamma=N$ together with the nonlinear equation~\eqref{eq:DF_SCF}. But the existence of minimizers plays no role in our study. 

Before studying its limit we show that $E^{\rm DF}(\kappa,\alpha,N)$ is a well-defined number. This turns out to be a simple consequence of the analysis in the previous section. 

\begin{lemma}[$\cE^{\rm DF}_{\kappa,\alpha}$ is bounded from below on positive energy states]\label{lem:coercive}
Let $0\leq\kappa<1$ and $0\leq\nu<\nu_0$. There exists a constant $\alpha_c=\alpha_c(\kappa,\nu)$ so that 
\begin{equation}
 \cE^{\rm DF}_{\kappa,\alpha}(\gamma)\geq \frac{c_{\kappa,\nu}}{2}\tr(\sqrt{-\Delta}\gamma)-\tr(\gamma)
 \label{eq:coercive}
\end{equation}
for all $0\leq \alpha\leq\alpha_c$ and every density matrix $0\leq\gamma=\gamma^*\leq 1$ satisfying the nonlinear constraint $\gamma P^-_{\kappa,\gamma}=0$ and such that $\alpha\tr(\gamma)\leq \nu$ and $\tr(\sqrt{-\Delta}\gamma)<\ii$.
In particular, $E^{\rm DF}(\kappa,\alpha,N)$ is well defined in~\eqref{eq:DF_energy}.
\end{lemma}

\begin{proof}
Since $D_{\kappa,\alpha\rho_\gamma}\gamma=|D_{\kappa,\alpha\rho_\gamma}|\gamma$, we have by~\eqref{eq:compare_Dirac}
\begin{align*}
\cE^{\rm DF}_{\kappa,\alpha}(\gamma)&=\tr D_{\kappa,\alpha\rho_\gamma}\gamma-\frac{\alpha}{2} \iint_{\R^3\times\R^3}\frac{\rho_\gamma(x)\rho_\gamma(y)}{|x-y|}\,dx\,dy-\tr(\gamma)\\
&\geq c_{\kappa,\nu}\tr\sqrt{1-\Delta}\gamma-\frac{\alpha}{2} \iint_{\R^3\times\R^3}\frac{\rho_\gamma(x)\rho_\gamma(y)}{|x-y|}\,dx\,dy-\tr(\gamma).
\end{align*}
The Lieb-Thirring inequality~\cite[Theorem 4.3]{LieSei-09} states that 
\begin{equation}
 \tr\sqrt{1-\Delta}\gamma\geq \tr\sqrt{-\Delta}\gamma\geq c_{\rm LT}\int_{\R^3}\rho_\gamma^{4/3}.
 \label{eq:Lieb-Thirring}
\end{equation}
On the other hand, the Hardy-Littlewood-Sobolev~\cite{LieLos-01} and H\"older inequalities give
\begin{equation}
 \iint_{\R^3\times\R^3}\frac{\rho_\gamma(x)\rho_\gamma(y)}{|x-y|}\,dx\,dy\leq c_{\rm HLS}\norm{\rho_\gamma}^2_{L^{6/5}(\R^3)}\leq c_{\rm HLS}(\tr\gamma)^{2/3}\int_{\R^3}\rho_\gamma^{4/3}
 \label{eq:estim_Direct}
\end{equation}
hence we conclude that 
$$\cE^{\rm DF}_{\kappa,\alpha}(\gamma)\geq \frac{c_{\kappa,\nu}}{2}\tr(\sqrt{-\Delta}\gamma)-\tr(\gamma)+\frac12\left(c_{\kappa,\nu}c_{\rm LT}-c_{\rm HLS}\alpha^{1/3}\nu^{2/3}\right)\int_{\R^3}\rho_\gamma^{4/3}$$
where the last term is non-negative for 
$$\alpha\leq\alpha_c:=\left(\frac{c_{\kappa,\nu}c_{\rm LT}}{c_{\rm HLS}\nu^{2/3}}\right)^3.$$
This concludes the proof. 
\end{proof}

%%%%%%%%%%%%%%%%%%%%%%%%%%%%%%%%%%%%%%%%%%%%%%%%%%%%%%%%%
\subsection{The Scott correction}
%%%%%%%%%%%%%%%%%%%%%%%%%%%%%%%%%%%%%%%%%%%%%%%%%%%%%%%%%

We are finally able to state the main theorem of this article. %We define the Thomas-Fermi energy
For a given constant $\lambda>0$, we define the Thomas-Fermi energy
\begin{multline}
e_{\rm TF}(\lambda)=
\min_{\substack{\rho\geq0\\ \int_{\R^3}\rho=1}}\bigg\{\frac3{10}(3\pi^2)^{\frac23} \int_{\R^3}\rho(x)^{\frac53}\,dx-\lambda\int_{\R^3}\frac{\rho(x)}{|x|}\,dx\\
+\frac12\iint_{\R^3\times\R^3}\frac{\rho(x)\rho(y)}{|x-y|}\,dx\,dy\bigg\}.
\label{eq:TF_energy_lambda}
\end{multline}
We note that $e_{\rm TF}(\lambda)$ is non-increasing and concave in $\lambda$. It is constant, equal to $e_{\rm TF}(1)<0$ for $\lambda\in(0;1]$. 
On the other hand, the Scott correction was defined in~\cite{HanSie-15} to be 
\begin{equation}
c_{\rm Scott}(\kappa):=\frac{\kappa^2}2+\sum_{n\geq1}\left\{\lambda_n\left(D_0-\frac{\kappa}{|x|}-1\right)-\lambda_n\left(-\frac{\Delta}{2}-\frac{\kappa}{|x|}\right)\right\},
\label{eq:def_Scott}
\end{equation}
where $\lambda_n$ are the eigenvalues in $[-2,0]$ of the operator in the parenthesis, repeated in case of multiplicity and arranged in increasing order. 

Our main result concerns the case of neutral atoms, as is classically considered for the Scott correction. This forces us to take $\nu=\kappa<\nu_0$. However, several parts of our approach apply to the case of general $\kappa<1$ and $\nu<\nu_0$, as we will see. 

\begin{theorem}[Scott correction in Dirac-Fock theory]\label{thm:Scott}
Let $0<\kappa< \nu_0$. Then we have
\begin{equation}
\lim_{\substack{N\to\ii\\ \alpha N\to\kappa}}\left|E^{\rm DF}(\kappa,\alpha,N)-e_{\rm TF}(1)\,\alpha^2N^{\frac73}-c_{\rm Scott}(\kappa)\right|=0.
\label{eq:Scott_main}
\end{equation}
\end{theorem}

Note that our energy is multiplied by $\alpha^2$ compared to~\cite{HanSie-15} and several other works on the subject. We have 
$$\alpha^2N^{\frac73}=\kappa^2 \,N^{\frac13}$$
hence the energy is of order $N^{1/3}$ in our regime, whereas the Scott correction is of order one. 

\begin{remark}
It is only because we rely on existing results, in particular from~\cite{HanSie-15}, that we need to impose $\kappa=\nu$ (neutral atoms). 
We believe that for all $0<\kappa<1$ and $0<\nu<\nu_0$, we have 
$$\lim_{\substack{N\to\ii\\ \alpha N\to\nu}}\left|E^{\rm DF}(\kappa,\alpha,N)-e_{\rm TF}(\kappa/\nu)\,\alpha^2N^{\frac73}-c_{\rm Scott}(\kappa)\right|=0.$$
That is, the result should apply to ions, as first considered by Bach in~\cite{Bach-89}. Several of our intermediate steps will actually be valid in this regime. 
\end{remark}

The rest of the paper is devoted to the proof of Theorems~\ref{thm:Dirac_gap} and~\ref{thm:Scott}.

%%%%%%%%%%%%%%%%%%%%%%%%%%%%%%%%%%%%%%%%%%%%%%%%%%%%%%%%%
%%%%%%%%%%%%%%%%%%%%%%%%%%%%%%%%%%%%%%%%%%%%%%%%%%%%%%%%% 
\section{Proof of Theorem~\ref{thm:Dirac_gap}} \label{sec:proof_Dirac_gap}
%%%%%%%%%%%%%%%%%%%%%%%%%%%%%%%%%%%%%%%%%%%%%%%%%%%%%%%%%
%%%%%%%%%%%%%%%%%%%%%%%%%%%%%%%%%%%%%%%%%%%%%%%%%%%%%%%%%

In this section we provide the proof of Theorem~\ref{thm:Dirac_gap} which states that 
$$ |D_{\kappa,\rho}|\geq c_{\kappa,\nu}|D_0|$$
as soon as $\int_{\R^3}\rho=\nu<\nu_0$ and $0\leq\kappa<1$. Our argument uses the Birman-Schwinger principle as in Nenciu's work~\cite{Nenciu-76} on the distinguished self-adjoint extensions of $D_\kappa$.

\begin{lemma}\label{lem:Birman-Schwinger}
Let $A$ be a self-adjoint operator such that $0\notin\sigma(A)$ and let $B$ be a positive, $A$-bounded operator, on a Hilbert space $\gH$. If $-1\notin \sigma(\sqrt{B}A^{-1}\sqrt{B})$ then  $0\notin\sigma(A+B)$ and the resolvent of $A+B$ is given by
\begin{equation}
 \frac{1}{A+B}=\frac{1}{A}-\frac{1}{A}\sqrt{B}\frac{1}{1+\sqrt{B}A^{-1}\sqrt{B}}\sqrt{B}\frac1A.
 \label{eq:resolvent_Nenciu}
\end{equation}
\end{lemma}

\begin{proof} First note that since $B$ is $A$-bounded, the operator $A+B$ is well defined on $D(A)$. If $-1\notin \sigma(\sqrt{B}A^{-1}\sqrt{B})$ then we can define the bounded operator
\begin{equation*}
R := \frac{1}{A}-\frac{1}{A}\sqrt{B}\frac{1}{1+\sqrt{B}A^{-1}\sqrt{B}}\sqrt{B}\frac1A.
\end{equation*}
It remains to verify that
\begin{align*}
(A+B)R &= 1 - \sqrt{B}\frac{1}{1+\sqrt{B}A^{-1}\sqrt{B}}\sqrt{B}\frac1A \\
&\qquad  + B  \frac{1}{A} - \sqrt{B}\sqrt{B}\frac{1}{A}\sqrt{B}\frac{1}{1+\sqrt{B}A^{-1}\sqrt{B}}\sqrt{B}\frac1A \\
	&= 1 -  B  \frac{1}{A} - \sqrt{B} \frac{1+\sqrt{B}A^{-1}\sqrt{B}}{1+\sqrt{B}A^{-1}\sqrt{B}}\sqrt{B}\frac1A = 1,
\end{align*}
which proves that $0 \notin \sigma(A+B)$ and (\ref{eq:resolvent_Nenciu}).
\end{proof}

The formula~\eqref{eq:resolvent_Nenciu} was used by Klaus with $A=D_0$ and $B=\kappa/|x|$ to prove the existence of the unique distinguished self-adjoint extension for $D_\kappa$. The critical value $\kappa=1$ arises from the fact that  
$$\norm{\frac1{|x|^{1/2}}\frac1{D_0}\frac1{|x|^{1/2}}}=1.$$
This relation has been conjectured by Nenciu~\cite{Nenciu-76} and was later proved by W\"ust~\cite{Wust-77} and Kato~\cite{Kato-83}. It has recently been rediscovered in~\cite[Thm.~1.3]{ArrDuoVeg-13}.

The lemma implies the following.

\begin{corollary}\label{cor:critical_nu}
Let $A$ be a self-adjoint operator such that $0\notin\sigma(A)$ and let $B$ be a positive bounded operator, on a Hilbert space $\gH$. Denote by
$$\lambda_c(B):=\min\sigma\left(\sqrt{B}\frac1A\sqrt{B}\right)$$
the minimum of the spectrum of $\sqrt{B}A^{-1}\sqrt{B}$ and
$$\nu_0(B)=\begin{cases}
-\frac1{\lambda_c(B)}&\text{if $\lambda_c(B)<0$}\\
+\ii&\text{otherwise.}
\end{cases}$$
Then 
$$0\notin\sigma(A+\nu B)$$
for all $0\leq \nu<\nu_0(B)$ whereas $0\in \sigma(A+\nu_0(B) B)$ if $\nu_0(B)<\ii$. 
\end{corollary}

\begin{proof}
For $\nu>0$ we have $1+\nu\sqrt{B}A^{-1}\sqrt{B}\geq 1+\nu\lambda_c(B)$ hence the result follows from Lemma~\ref{lem:Birman-Schwinger}.
\end{proof}

Now we go back to our Dirac operator. Note that by charge conjugation invariance, the spectrum of the operator 
$$M_\rho=\sqrt{\rho\ast|x|^{-1}}\frac{1}{D_0}\sqrt{\rho\ast|x|^{-1}}$$
is symmetric with respect to the origin. When $\sqrt{\rho}\in H^{1/2}(\R^3,\C^4)$, $M_\rho$ is a compact operator. Corollary~\ref{cor:critical_nu} implies that the critical value $\nu_0(\rho)$ at which the largest eigenvalue of $D_0+\nu\rho\ast|x|^{-1}$ crosses the origin is given by
$$-\frac1{\nu_0(\rho)}=-\|M_\rho\|=\inf_{\Psi\in L^2(\R^3,\C^4)}\frac{\pscal{\Psi,M_\rho\Psi}}{\|\Psi\|^2}.$$
In particular,
$$-\frac1{\nu_0}=\inf_{\substack{\rho\geq0\\ \sqrt{\rho}\in H^{1/2}(\R^3)}}\inf_{\Psi\in L^2(\R^3,\C^4)}\frac{\pscal{\Psi,M_\rho\Psi}}{\|\Psi\|^2}$$
as used already in~\cite{EstLewSer-19b_ppt}.

Now we look at the operator 
$$D_{\kappa,\rho}=D_\kappa+\nu\rho\ast\frac{1}{|x|}$$
and recall that $D_\kappa$ is invertible for $0\leq\kappa<1$. The previous theory tells us that for a probability density $\rho$ with $\sqrt{\rho}\in H^{1/2}(\R^3)$, no eigenvalue will cross 0 for $\nu<\nu_0(\kappa,\rho)$ given by
$$-\frac{1}{\nu_0(\kappa,\rho)}=\inf_{\Psi\in L^2(\R^3,\C^4)}\frac{\pscal{\sqrt{\rho\ast|x|^{-1}}\Psi,\frac1{D_\kappa}\sqrt{\rho\ast|x|^{-1}}\Psi}}{\|\Psi\|^2}.$$

In order to prove a lower bound on $|D_{\kappa,\rho}|$, we use Formula~\eqref{eq:resolvent_Nenciu}. We write
\begin{align*}
|D_0|^{1/2}\frac{1}{|D_{\kappa,\rho}|^{1/2}}&=|D_0|^{1/2}\frac{1}{D_{\kappa,\rho}}|D_{\kappa,\rho}|^{1/2}U_{\kappa,\rho}\\
&=\bigg(|D_0|^{1/2}\frac{1}{D_\kappa}|D_{\kappa,\rho}|^{1/2}\\
&\qquad\qquad+|D_0|^{1/2}\frac{1}{D_\kappa}\sqrt{V_\rho}\frac1{1+M_{\kappa,\rho}}\sqrt{V_\rho}\frac1{D_\kappa}|D_{\kappa,\rho}|^{1/2}\bigg)U_{\kappa,\rho},
\end{align*}
where $U_{\kappa,\rho}={\rm sgn}(D_{\kappa,\rho})$ is a unitary operator, $V_\rho=\rho\ast|x|^{-1}$ and $M_{\kappa,\rho}=\sqrt{V_\rho}D_\kappa^{-1}\sqrt{V_\rho}$. We have 
$$\|(1+M_{\kappa,\rho})^{-1}\|\leq \frac{\nu_0}{\nu_0-\nu}.$$
This gives 
\begin{multline}
\norm{|D_0|^{1/2}\frac{1}{|D_{\kappa,\rho}|^{1/2}}}\leq \norm{|D_0|^{1/2}\frac1{|D_\kappa|^{1/2}}}\norm{\frac1{|D_\kappa|^{1/2}}|D_{\kappa,\rho}|^{1/2}}\times\\ \times\bigg(1+\frac{\nu_0}{\nu_0-\nu}\norm{\sqrt{V_\rho}\frac1{|D_\kappa|^{1/2}}}^2\bigg). 
\label{eq:estim_norm_D_kappa_rho}
\end{multline}
The best constant $c_\kappa$ such that $c_\kappa|D_0|\leq|D_\kappa|$ is exactly given by 
$$\frac1{c_\kappa}=\norm{\frac1{|D_\kappa|^{1/2}}|D_0|\frac1{|D_\kappa|^{1/2}}}=\norm{|D_0|^{1/2}\frac1{|D_\kappa|^{1/2}}}^2.$$
From the Hardy-Kato inequality~\eqref{eq:Kato}, we have 
$$\norm{\sqrt{V_\rho}\frac1{|D_\kappa|^{1/2}}}^2=\norm{\frac1{|D_\kappa|^{1/2}}V_\rho\frac1{|D_\kappa|^{1/2}}}\leq\frac{\pi\nu}2 \norm{\frac1{|D_\kappa|^{1/2}}|D_0|\frac1{|D_\kappa|^{1/2}}}=\frac{\pi\nu}{2c_\kappa}.$$
On the other hand we have by~\eqref{eq:upper_bd_D}
\begin{align}
\norm{\frac1{|D_\kappa|^{1/2}}|D_{\kappa,\rho}|^{1/2}}&\leq \norm{\frac1{|D_\kappa|^{1/2}}|D_0|^{1/2}}\norm{\frac1{|D_0|^{1/2}}|D_{\kappa,\rho}|^{1/2}}\nn\\
&\leq \sqrt{\frac{1+2(\kappa + \nu)}{c_\kappa}}.\label{eq:not_so_good_bound}
\end{align}
Inserting in~\eqref{eq:estim_norm_D_kappa_rho}, we obtain 
$$\norm{|D_0|^{1/2}\frac{1}{|D_{\kappa,\rho}|^{1/2}}}^2\leq \frac{1 + 2(\nu+\kappa)}{c_\kappa^2}\bigg(1+\frac{\pi}{2c_\kappa}\frac{\nu_0\nu}{\nu_0-\nu}\bigg)^2,$$
which concludes the proof of Theorem~\ref{thm:Dirac_gap}.\qed

\begin{remark}
Our bound~\eqref{eq:not_so_good_bound} is far from optimal, in particular when $\nu=0$. This is why we obtain $c_{\kappa,0}\neq c_\kappa$.
\end{remark}

%%%%%%%%%%%%%%%%%%%%%%%%%%%%%%%%%%%%%%%%%%%%%%%%%%%%%%%%%
%%%%%%%%%%%%%%%%%%%%%%%%%%%%%%%%%%%%%%%%%%%%%%%%%%%%%%%%% 
\section{A bound on differences of spectral projections}
%%%%%%%%%%%%%%%%%%%%%%%%%%%%%%%%%%%%%%%%%%%%%%%%%%%%%%%%%
%%%%%%%%%%%%%%%%%%%%%%%%%%%%%%%%%%%%%%%%%%%%%%%%%%%%%%%%%

In this section we prove the following estimate on the difference of two spectral projections, when the electric field has a small energy, first in Hilbert-Schmidt norm and then in the Schatten space $\gS^6$. We recall that $\norm{A}_{\gS^p}:=(\tr|A|^p)^{1/p}$.

\begin{proposition}[Schatten class estimates on differences of projections]\label{prop:difference_projections}
Let $0\leq\kappa<1$ and $\eps>0$. There exists a constant $B=B(\kappa)$ and a constant $C=C(\kappa,\eps)$  such that 
\begin{multline}
 \norm{|D_\kappa|^{\frac12}\Big(\1(D_\kappa+V\leq0)-\1(D_\kappa\leq0)\Big)|D_\kappa|^{-\frac12}}_{\gS^2}\\
 \leq B\left(1+\norm{\nabla V}_{L^3(\R^3)}\right)\norm{\nabla V}_{L^2(\R^3)}
 \label{eq:estim_proj_S2}
\end{multline}
and 
\begin{equation}
 \norm{|D_\kappa|^{\frac12}\Big(\1(D_\kappa+V\leq0)-\1(D_\kappa\leq0)\Big)|D_\kappa|^{-\eps}}_{\gS^6}\leq C\norm{\nabla V}_{L^2(\R^3)}
 \label{eq:estim_proj_S6}
\end{equation}
for every $V\in L^6(\R^3)$ such that $\norm{V}_{L^6(\R^3)}\leq 1/B$. 

This implies
\begin{multline}
 \norm{|D_\kappa|^{\frac12}\Big(\1(D_\kappa+V\leq0)-\1(D_\kappa\leq0)\Big)|D_\kappa|^{-\frac{6-p}{8}-\eps}}_{\gS^p}\\
 \leq C^{\frac{3p-6}{2p}}B^{\frac{6-p}{2p}}\left(1+\norm{\nabla V}_{L^3(\R^3)}\right)^{\frac{6-p}{2p}}\norm{\nabla V}_{L^2(\R^3)}
 \label{eq:estim_proj_Sp}
\end{multline}
for all $2\leq p\leq6$, for every $\|V\|_{L^6(\mathbb{R}^3)} \leq 1/B$.
\end{proposition}

In the lemma we can replace $|D_\kappa|^{1/2}$ by $|D_0|^{1/2}$ everywhere since those are comparable.

\begin{proof}
In the whole proof we denote by $C$ a generic constant whose value can change from line to line, but which only depends on $\kappa$ and $\eps$. 
For $0\leq\kappa<1$ we have by~\eqref{eq:compare_Dirac_kappa}
\begin{equation}
\norm{ V\frac{1}{D_\kappa}}_{\gS^6}\leq C\norm{V\frac{1}{|D_0|^{\frac12+\eps}}}_{\gS^6}\leq C\norm{V}_{L^6(\R^3)}
 \label{eq:Rellich-Kato_V}
\end{equation}
for every $\eps<\sqrt{1-\kappa^2}$, where in the last inequality we have used the Kato-Seiler-Simon inequality~\cite{Simon-79}
\begin{equation}
\norm{f(x)g(-i\nabla)}_{\gS^p}\leq (2\pi)^{-\frac{d}{p}}\norm{f}_{L^p(\R^d)}\norm{g}_{L^p(\R^d)},\qquad \forall p\geq2.
\label{eq:KSS}
\end{equation}
By the Rellich-Kato theorem, this proves that when $\norm{V}_{L^6}$ is small enough, $D_\kappa+V$ is self-adjoint on the same domain as $D_\kappa$, with 
\begin{equation}
 \frac1C \big(D_\kappa+V)^2\leq (D_\kappa)^2\leq C(D_\kappa+V)^2
 \label{eq:compare_D_kappa_V}
\end{equation}
for some constant $C$ depending only on the maximal allowed value of $\|V\|_{L^6}$. In particular, $D_\kappa+V$ has a gap around the origin. In the rest of the proof we always assume that $\|V\|_{L^6}$ is small enough.

Throughout the proof we denote for simplicity $P_{\kappa,V}^\pm:=\1_{\R^\pm}(D_0+V)$. 

\subsubsection*{$\bullet$ Proof of the $\gS^6$ estimate (\ref{eq:estim_proj_S6})}
We use Stone's formula for spectral projections and the resolvent formula to express the difference as
$$P_{\kappa,V}^+-P_{\kappa}^+=-\frac{1}{2\pi}\int_\R\frac{1}{D_\kappa+i\eta}V \frac{1}{D_\kappa+V+i\eta}\,d\eta.$$
Hence
\begin{multline}
\norm{|D_\kappa|^{\frac12}(P_{\kappa,V}^+-P_{\kappa}^+)|D_\kappa|^{-\eps}}_{\gS^6}\\
\leq \frac{1}{2\pi}\int_\R\norm{\frac{|D_\kappa|^{\frac12}}{D_\kappa+i\eta}V \frac{1}{D_\kappa+V+i\eta}\frac1{|D_\kappa|^{\eps}}}_{\gS^6}d\eta. 
\label{eq:estim_diff_proj_1}
\end{multline}
Writing
\begin{multline*}
\frac{1}{D_\kappa+V+i\eta}\frac1{|D_\kappa|^{\eps}}=\frac1{|D_\kappa|^{\eps}} \left(|D_{\kappa}|^\eps\frac{1}{|D_{\kappa,V}|^\eps}\right) \frac{1}{D_\kappa+V+i\eta}\left(|D_{\kappa,V}|^\eps\frac1{|D_\kappa|^\eps}\right) 
\end{multline*}
inserting in~\eqref{eq:estim_diff_proj_1} and using~\eqref{eq:compare_D_kappa_V}, we find
\begin{align*}
&\norm{\frac{|D_\kappa|^{\frac12}}{D_\kappa+i\eta}V \frac{1}{D_\kappa+V+i\eta}\frac1{|D_\kappa|^{\eps}}}_{\gS^6}\\
&\qquad \qquad \leq \frac{C}{\pscal{\eta}}\norm{\frac{1}{|D_\kappa|^{\frac12}+\pscal{\eta}^{\frac12}}V \frac{1}{|D_\kappa|^\eps}}_{\gS^6}\\
&\qquad \qquad \leq \frac{C}{\pscal{\eta}}\norm{\frac{1}{|D_0|^{\frac12}+\pscal{\eta}^{\frac12}}V \frac{1}{|D_0|^\eps}}_{\gS^6}\leq \frac{C\norm{V}_{L^6}}{\pscal{\eta}^{1+\eps}}.
\end{align*}
Here $\pscal{\eta}=\sqrt{1+\eta^2}$ is the Japanese bracket. We have used the Kato-Seiler-Simon inequality~\eqref{eq:KSS} and the fact that $|D_\kappa|$ is comparable to $|D_0|$ in the quadratic form sense. 
We obtain~\eqref{eq:estim_proj_S6} after integrating over $\eta$ and using the Sobolev inequality $\|V\|_{L^6}\leq C\|\nabla V\|_{L^2}$.

\subsubsection*{$\bullet$ Proof of the Hilbert-Schmidt estimate (\ref{eq:estim_proj_S2})}
The proof of (\ref{eq:estim_proj_S2}) is much more involved. We start by iterating the resolvent formula twice to obtain
\begin{align}
P_{\kappa,V}^+-P_{\kappa}^+&=-\frac{1}{2\pi}\int_\R\frac{1}{D_\kappa+i\eta}V \frac{1}{D_\kappa+i\eta}\,d\eta\nn\\
&\qquad+\frac{1}{2\pi}\int_\R\frac{1}{D_\kappa+i\eta}V\frac{1}{D_\kappa+i\eta}V \frac{1}{D_\kappa+i\eta}\,d\eta\nn\\
&\qquad-\frac{1}{2\pi}\int_\R\frac{1}{D_\kappa+i\eta}V\frac{1}{D_\kappa+i\eta}V\frac{1}{D_\kappa+i\eta}V \frac{1}{D_\kappa+V+i\eta}\,d\eta.
\label{eq:decomp_difference_proj}
\end{align}
Since $D_\kappa+V$ is comparable to $D_\kappa$ by~\eqref{eq:compare_D_kappa_V} and $|D_\kappa|^{1/2}$ is comparable to $|D_0|^{1/2}$, the last term can be bounded by H\"older's inequality in Schatten spaces by
\begin{align*}
&\norm{|D_\kappa|^{\frac12}\int_\R\frac{1}{D_\kappa+i\eta}V\frac{1}{D_\kappa+i\eta}V\frac{1}{D_\kappa+i\eta}V \frac{1}{D_\kappa+V+i\eta}\,d\eta}_{\gS^2}\\
&\qquad\leq C\int_{\R} \frac{d\eta}{\pscal{\eta}^{\frac12}}\norm{\frac{1}{|D_0|^{\frac12}+\pscal{\eta}^{\frac12}}V\frac{1}{|D_0|^{\frac12}+\pscal{\eta}^{\frac12}}}_{\gS^6}^3\\
&\qquad\leq C\norm{V}_{L^6}^3\int_{\R} \frac{d\eta}{\pscal{\eta}^2}. 
\end{align*}
It is here not necessary to use the operator $|D_\kappa|^{-1/2}$ on the right side. It therefore remains to estimate the first two terms in~\eqref{eq:decomp_difference_proj}.

We start with the second term in~\eqref{eq:decomp_difference_proj}. Using that 
$$\int_\R\frac{1}{(D_\kappa+i\eta)^3}\,d\eta=0$$
by Cauchy's formula, we have 
\begin{multline*}
\int_\R\frac{1}{D_\kappa+i\eta}V\frac{1}{D_\kappa+i\eta}V \frac{1}{D_\kappa+i\eta}\,d\eta \\
=\int_\R\frac{1}{D_\kappa+i\eta}V\frac{1}{D_\kappa+i\eta}\left[V, \frac{1}{D_\kappa+i\eta}\right]\,d\eta
+\int_\R\left[\frac{1}{D_\kappa+i\eta},V\right]\frac{1}{(D_\kappa+i\eta)^2} V\,d\eta.
\end{multline*}
Inserting then 
$$\left[V, \frac{1}{D_\kappa+i\eta}\right]=\frac{1}{D_\kappa+i\eta}\left[D_\kappa , V\right]\frac{1}{D_\kappa+i\eta}=-i\frac{1}{D_\kappa+i\eta}\alp\cdot \nabla V\frac{1}{D_\kappa+i\eta}$$
we obtain
\begin{multline*}
\int_\R\frac{1}{D_\kappa+i\eta}V\frac{1}{D_\kappa+i\eta}V \frac{1}{D_\kappa+i\eta}\,d\eta \\
=-i\int_\R\frac{1}{D_\kappa+i\eta}V\frac{1}{(D_\kappa+i\eta)^2}\alp\cdot \nabla V\frac{1}{D_\kappa+i\eta}\,d\eta\\
+i\int_\R \frac{1}{D_\kappa+i\eta}\alp\cdot \nabla V\frac{1}{(D_\kappa+i\eta)^3} V\,d\eta.
\end{multline*}
In order to estimate this term, we use that $|D_\kappa|^{1+\eps}\geq c|D_0|^{1+\eps'}$ for every $\eps'<\min(\eps,2\eps\sqrt{1-\kappa^2})$, see~\eqref{eq:lower_bound_Dirac_Coulomb_powers_eta}. This gives
\begin{align*}
&\norm{|D_\kappa|^{\frac12}\int_\R\frac{1}{D_\kappa+i\eta}V\frac{1}{D_\kappa+i\eta}V \frac{1}{D_\kappa+i\eta}\,d\eta\;|D_\kappa|^{-\frac12}}_{\gS^2}\\
&\qquad\leq C\int_\R\norm{\frac{1}{|D_0|^{\frac12}}V\frac{1}{|D_0|^{\frac12}}}_{\gS^6}\norm{\frac{1}{|D_0|^{\frac12}}\alp\cdot \nabla V\frac{1}{|D_0|^{\frac12+\eps'}}}_{\gS^3}\,\frac{d\eta}{\pscal{\eta}^{2-\eps}}\\
&\qquad\leq C\norm{V}_{L^6}\norm{\nabla V}_{L^3}.
\end{align*}
This gives rise to the term $\norm{V}_{L^6}\norm{\nabla V}_{L^3}$ in our estimate~\eqref{eq:estim_proj_S2}.

Finally, we deal with the first term in~\eqref{eq:decomp_difference_proj}. If we had $D_0$ in place of $D_\kappa$, the result would follow directly from the Kato-Seiler-Simon inequality. The difficulty here is that high powers of $D_\kappa$ are not comparable with $D_0$ when $\kappa$ is close to $1$. So we compute the difference exactly. We insert the resolvent formula
\begin{align*}
\frac{1}{\Dk + i\eta} &= \frac{1}{D_0 + i\eta} -\kappa\frac{1}{D_0 + i\eta}\frac{1}{|x|} \frac{1}{\Dk + i\eta}\\ % \label{res_id_0_k_1}\\
&=\frac{1}{D_0 + i\eta} -\kappa\frac{1}{\Dk + i\eta}\frac{1}{|x|} \frac{1}{D_0 + i\eta} %\label{res_id_0_k_2}
\end{align*}
in the first term in (\ref{eq:decomp_difference_proj}) and we obtain the rather lengthy formula
\begin{align}
&\int_{\mathbb{R}} |\Dk|^{\frac12} \frac{1}{D_\kappa + i\eta} V  \frac{1}{D_\kappa + i\eta}|\Dk|^{-\frac12}d\eta\nn\\
&\ =\int_{\mathbb{R}} |\Dk|^{\frac12} \frac{1}{D_0 + i\eta} V  \frac{1}{D_0 + i\eta}|\Dk|^{-\frac12}d\eta \nn\\
&\quad  - \kappa\int_{\mathbb{R}} |\Dk|^{\frac12} \left(\frac{1}{D_0+ i\eta}\frac{1}{|x|} \frac{1}{D_0 + i\eta} V  \frac{1}{D_0 + i\eta}+\text{h.c.}\right)|\Dk|^{-\frac12}d\eta \nn\\
&\quad + \kappa^2\int_{\mathbb{R}} |\Dk|^{\frac12}\left( \frac{1}{D_\kappa + i\eta}\frac{1}{|x|} \frac{1}{D_0 + i\eta}\frac{1}{|x|} \frac{1}{D_0 + i\eta} V  \frac{1}{D_0 + i\eta}+\text{h.c.}\right)|\Dk|^{-\frac12}d\eta \nn\\
&\quad +\kappa^2\int_{\mathbb{R}} |\Dk|^{\frac12}\frac{1}{D_\kappa + i\eta}\frac{1}{|x|} \frac{1}{D_0 + i\eta} V  \frac{1}{D_0 + i\eta}\frac{1}{|x|} \frac{1}{\Dk + i\eta}|\Dk|^{-\frac12}d\eta.\label{eq:expand_1st_order_kappa}
\end{align}
In order to estimate the last two terms we can use that for $s>0$, we have 
$$\frac{1}{|D_0|^{\frac12}}\frac{1}{|x|}\frac{1}{|D_0|^{\frac12+s}}\in\gS^3_w$$
by Cwikel's inequality~\cite{Simon-79}. In particular, we deduce that 
$$\frac{1}{|D_0|^{\frac12}}\frac{1}{|x|}\frac{1}{|D_0|^{\frac12+s}}\in \gS^p$$
for all $p>3$ and all $s>0$. For instance we can control the last term by
\begin{align*}
&\norm{|\Dk|^{\frac12}\frac{1}{D_\kappa + i\eta}\frac{1}{|x|} \frac{1}{D_0 + i\eta} V  \frac{1}{D_0 + i\eta}\frac{1}{|x|} \frac{1}{\Dk + i\eta}|\Dk|^{-\frac12}}_{\gS^2}\\
&\qquad \leq\frac{C}{\pscal{\eta}}\norm{\frac{1}{|D_0|^{\frac12}}\frac{1}{|x|}\frac{1}{|D_0|^{\frac12+s}}}_{\gS^4}^2\norm{\frac{1}{(|D_0|+\pscal{\eta})^{\frac12-s}}V\frac{1}{(|D_0|+\pscal{\eta})^{\frac12-s}}}_{\gS^6}\\
&\qquad \leq\frac{C}{\pscal{\eta}^{\frac32-6s}}\norm{V}_{L^6}
\end{align*}
which is integrable over $\eta$ for $s>0$ small enough. The argument is the same for the other term of order $\kappa^2$. 
On the other hand, for the first term in~\eqref{eq:expand_1st_order_kappa}, we use that
$$\int_{\mathbb{R}} \frac{1}{(D_0+ i\eta)^2}\,d\eta = 0$$
and insert one commutator, which yields
\begin{multline*}
\norm{\int_{\mathbb{R}} |\Dk|^{\frac12} \frac{1}{D_0+ i\eta} V  \frac{1}{D_0+ i\eta}|\Dk|^{-\frac12}d\eta}_{\gS^2}\\
=\norm{\int_{\mathbb{R}} |\Dk|^{\frac12} \frac{1}{(D_0 + i\eta)^2} \alp\cdot \nabla V  \frac{1}{D_0 + i\eta}|\Dk|^{-\frac12}d\eta}_{\gS^2}\leq C\norm{\nabla V}_{L^2}.
\end{multline*}

It remains to estimate the second term in~\eqref{eq:expand_1st_order_kappa}
$$\int_{\mathbb{R}} |\Dk|^{\frac12} \frac{1}{D_0+ i\eta}\frac{1}{|x|} \frac{1}{D_0 + i\eta} V  \frac{1}{D_0 + i\eta}|\Dk|^{-\frac12}d\eta .$$
This is the most difficult since $V\in L^6$ and $1/|x|$ only yields an operator in $\gS^{3}_w$, by Cwikel's inequality. The idea here is to split
$$\frac{1}{|x|}=\frac{\chi(x)}{|x|}+\frac{1-\chi(x)}{|x|}$$
where $\chi\in C^\ii_c$ is equal to 1 in a neighborhood of the origin. The term involving $\chi/|x|$ is easily handled using that $\chi/|x|\in L^{p}$ for all $p<3$, hence 
$$\frac{1}{|D_0|^{\frac12}}\frac{\chi}{|x|}\frac{1}{|D_0|^{\frac12+s}}\in \gS^p\subset \gS^3$$
for $s>0$. We can then write
\begin{multline*}
\int_\R\norm{|\Dk|^{\frac12} \frac{1}{D_0+ i\eta}\frac{\chi}{|x|} \frac{1}{D_0 + i\eta} V  \frac{1}{D_0 + i\eta}|\Dk|^{-\frac12}}_{\gS^2} \,d\eta\\
\leq \int_\R\frac{C}{\pscal{\eta}}\norm{\frac{1}{|D_0|^{\frac12}}\frac{\chi}{|x|} \frac{1}{|D_0|^{\frac12+s}}}_{\gS^3}\norm{\frac{1}{(|D_0|+\pscal{\eta})^{\frac12-s}}V\frac{1}{|D_0|^{\frac12}}}_{\gS^6}\,d\eta\leq C\norm{V}_{L^6}.
\end{multline*}

The term with $(1-\chi)/|x|$ is treated exactly as we did before for the quadratic term in $V$. Namely, we write
\begin{multline*}
|\Dk|^{\frac12}\int_{\mathbb{R}}  \frac{1}{D_0+ i\eta}\frac{1-\chi}{|x|} \frac{1}{D_0 + i\eta} V  \frac{1}{D_0 + i\eta}d\eta|\Dk|^{-\frac12}\\
=-i|\Dk|^{\frac12}\int_\R\frac{1}{D_0+i\eta}\frac{1-\chi}{|x|}\frac{1}{(D_0+i\eta)^2}\alp\cdot \nabla V\frac{1}{D_0+i\eta}\,d\eta|\Dk|^{-\frac12}\\
+i|\Dk|^{\frac12}\int_\R \frac{1}{D_0+i\eta}\alp\cdot \nabla \left(\frac{1-\chi}{|x|}\right)\frac{1}{(D_0+i\eta)^3} V\,d\eta|\Dk|^{-\frac12}
\end{multline*}
Now it suffices to use that $(1-\chi)/|x|$ is bounded and that its gradient is in $L^3$ to conclude. 

Our final estimate takes the form
$$C(1+\norm{\nabla V}_{L^3})\norm{V}_{L^6}+\norm{\nabla V}_{L^2}$$ 
for $\norm{V}_{L^6}$ small enough. We obtain the stated inequality~\eqref{eq:estim_proj_S2}. 

Finally, the last inequality~\eqref{eq:estim_proj_Sp} follows by complex interpolation. This concludes the proof of Proposition~\ref{prop:difference_projections}.
\end{proof}

%%%%%%%%%%%%%%%%%%%%%%%%%%%%%%%%%%%%%%%%%%%%%%%%%
%%%%%%%%%%%%%%%%%%%%%%%%%%%%%%%%%%%%%%%%%%%%%%%%%
\section{Proof of Theorem~\ref{thm:Scott}} 
%%%%%%%%%%%%%%%%%%%%%%%%%%%%%%%%%%%%%%%%%%%%%%%%%
%%%%%%%%%%%%%%%%%%%%%%%%%%%%%%%%%%%%%%%%%%%%%%%%%

In the whole argument we fix $0<\kappa<1$ and $0<\nu<\nu_0$ and we assume that $\alpha$ is small enough. Only when required we will impose $\kappa=\nu$. In order to simplify our writing we change notation and denote by
$$D_{\kappa,\gamma}:=D_0-\frac{\kappa}{|x|}+\alpha\rho_\gamma\ast\frac{1}{|x|}$$
the mean-field operator and by
$$V_\gamma:=\alpha\rho_\gamma\ast\frac1{|x|}$$
the corresponding mean-field operator. We recall that
$$P_{\kappa,\gamma}^{\pm}=\1_{\R_\pm}(D_{\kappa,\gamma})$$
are the associated spectral projection. Finally, we denote by
$$D(f,f):=\iint_{\R^3\times\R^3}\frac{f(x)f(y)}{|x-y|}\,dx\,dy=\frac1{4\pi}\int_{\R^3}\frac{|\widehat{f}(k)|^2}{|k|^2}\,dk$$
the Coulomb energy.

\subsection{Lower bound}

In this section we prove the following result.

\begin{proposition}[Lower bound in terms of the Dirac-Coulomb projected Dirac-Fock]\label{prop:lower_bound}
Let $0<\kappa<1$ and $0<\nu<\nu_0$. Then we have for a constant $C$ depending on $\kappa$ and $\nu$
\begin{equation}
E^{\rm DF}(\kappa,\alpha,N)\geq \inf_{\substack{0\leq\gamma\leq 1\\ P_\kappa^-\gamma=0\\ \tr(\gamma)\leq N}}\left\{\tr(D_\kappa-1)\gamma+\frac{\alpha}{2}D(\rho_\gamma,\rho_\gamma)\right\}-\frac{C}{N^{\frac1{15}}}
\end{equation}
for all $\alpha N\leq \nu$ and $\alpha$ small enough. 
\end{proposition}

\begin{proof}
From Lemma~\ref{lem:coercive} we have 
$$0\geq E^{\rm DF}(\kappa,\alpha,N)\geq -CN.$$ 
Note that $E^{\rm DF}(\kappa,\alpha,N)\leq0$, since one can take $\gamma=0$ in the variational principle~\eqref{eq:DF_energy}. 
We use a kind of boot-strap argument, showing first a lower bound of the order $-CN^{1/3}$ before getting lower order errors.

\subsubsection*{$\bullet$ Proof that $E^{\rm DF}(\kappa,\alpha,N)\geq -CN^{1/3}$}
Let $\gamma_N$ be an approximate minimizer for $E^{\rm DF}(\kappa,\alpha,N)$. Then by~\eqref{eq:coercive}
$$\tr\sqrt{-\Delta}\gamma_N\leq CN.$$
By~\eqref{eq:estim_Direct} and the Lieb-Thirring inequality~\eqref{eq:Lieb-Thirring}, we have for all density matrices $\gamma$
\begin{equation}
  \iint_{\R^3\times\R^3}\frac{\rho_\gamma(x)\rho_\gamma(y)}{|x-y|}\,dx\,dy\leq C\big(\tr\gamma\big)^{2/3}\tr(\sqrt{-\Delta}\gamma).
 \label{eq:estim_Direct_kinetic}
\end{equation}
In particular, the direct term in the Dirac-Fock energy satisfies
\begin{equation}
 \alpha\iint_{\R^3\times\R^3}\frac{\rho_{\gamma_N}(x)\rho_{\gamma_N}(y)}{|x-y|}\,dx\,dy\leq CN^{2/3}.
 \label{eq:estim_Direct_N23}
\end{equation}
Going back to the Dirac-Fock energy and using that $E^{\rm DF}(\kappa,\alpha,N)\leq0$, we find 
$$\tr(D_{\kappa,\gamma_N}-1)\gamma_N=\tr(|D_{\kappa,\gamma_N}|-1)\gamma_N\leq CN^{2/3}.$$
Now we replace $D_{\kappa,\gamma_N}$  by $D_\kappa$. We have 
\begin{align}
\tr D_{\kappa,\gamma_N}\gamma_N&=\tr P_{\kappa,\gamma_N}^+(D_{\kappa}+V_{\gamma_N}) P_{\kappa,\gamma_N}^+\gamma_N\nn\\
&=\tr P_{\kappa}^+D_{\kappa}P_{\kappa}^+\gamma_N+\alpha D(\rho_{\gamma_N},\rho_{\gamma_N})\nn\\
&\qquad+\tr (P_{\kappa,\gamma_N}^+-P_\kappa^+)D_{\kappa}(P_{\kappa,\gamma_N}^+-P_\kappa^+)\gamma_N\label{eq:decomp_kinetic}
\end{align}
since
$$\tr P_\kappa^+D_{\kappa}(P_{\kappa,\gamma_N}^+-P_\kappa^+)\gamma_N=-\tr P_\kappa^+D_{\kappa}(P_{\kappa,\gamma_N}^--P_\kappa^-)\gamma_N=0.$$
Note that by~\eqref{eq:estim_Direct_N23}
$$\norm{\alpha\rho_{\gamma_N}\ast\frac1{|x|}}_{L^6}\leq C\alpha\, D(\rho_{\gamma_N},\rho_{\gamma_N})^{\frac12}\leq C\sqrt{\alpha}N^{\frac13}\leq C\frac{\sqrt{\nu}}{N^{\frac16}}\to0.$$
Hence we may apply Proposition~\ref{prop:difference_projections}. In addition, we have
$$\norm{\alpha\nabla \left(\rho_{\gamma_N}\ast\frac1{|x|}\right)}_{L^3}\leq \alpha\norm{\rho_{\gamma_N}\ast\frac1{|x|^2}}_{L^3}\leq C\alpha\norm{\rho_{\gamma_N}}_{L^{3/2}}$$
by the Hardy-Littlewood-Sobolev inequality. Recall the Hoffmann-Ostenhof inequality
\begin{equation}
\tr\sqrt{-\Delta}\gamma\geq \pscal{\sqrt{\rho_\gamma},\sqrt{-\Delta}\sqrt{\rho_\gamma}}
\label{eq:Hoff_Ost}
\end{equation}
which follows from the convexity of fractional gradients~\cite[Thm.~7.13]{LieLos-01}. Using the Sobolev inequality and~\eqref{eq:Hoff_Ost}, we obtain
$$\norm{\rho_{\gamma_N}}_{L^{3/2}}\leq C\pscal{\sqrt{\rho_{\gamma_N}},\sqrt{-\Delta}\sqrt{\rho_{\gamma_N}}}\leq C\tr\sqrt{-\Delta}\gamma_N\leq CN.$$
Hence 
$$\norm{\alpha\nabla \left(\rho_{\gamma_N}\ast\frac1{|x|}\right)}_{L^3}\leq C\alpha N\leq C\nu$$
is uniformly bounded. By H\"older's inequality in Schatten spaces and Proposition~\ref{prop:difference_projections} with $3/10>1/8$ we can now bound
\begin{align}
&\bigg|\tr (P_{\kappa,\gamma_N}^+-P_\kappa^+)D_{\kappa}(P_{\kappa,\gamma_N}^+-P_\kappa^+)\gamma_N\bigg|\nn\\
&\qquad\qquad\leq \norm{|D_\kappa|^{\frac12}(P_{\kappa,\gamma_N}^+-P_\kappa^+)|D_\kappa|^{-\frac3{10}}}^2_{\gS^5}\norm{|D_\kappa|^{\frac{3}{10}}\gamma_N|D_\kappa|^{\frac{3}{10}}}_{\gS^{5/3}}\nn\\
&\qquad\qquad\leq CN^{\frac35} \alpha^2 D(\rho_{\gamma_N},\rho_{\gamma_N}).\label{eq:estim_error_proj}
\end{align}
We have used here that 
\begin{align*}
\norm{|D_\kappa|^{\frac{3}{10}}\gamma_N|D_\kappa|^{\frac{3}{10}}}_{\gS^{5/3}}^{\frac53}&=\tr\Big(|D_\kappa|^{\frac{3}{10}}\gamma_N|D_\kappa|^{\frac{3}{10}}\Big)^{\frac53}\\
&\leq \tr|D_\kappa|^{\frac12}\gamma_N^{\frac53}|D_\kappa|^{\frac12}\leq \tr|D_\kappa|^{\frac12}\gamma_N|D_\kappa|^{\frac12}\leq CN 
\end{align*}
since $0\leq\gamma_N\leq1$, and by the Araki-Lieb-Thirring inequality~\cite{LieThi-76,LieSei-09}. The same argument as for~\eqref{eq:estim_error_proj} implies also that
$$\bigg|\tr (P_{\kappa,\gamma_N}^+-P_\kappa^+)(P_{\kappa,\gamma_N}^+-P_\kappa^+)\gamma_N\bigg|\leq CN^{\frac35} \alpha^2 D(\rho_{\gamma_N},\rho_{\gamma_N}).$$
As a conclusion we have proved the lower bound
\begin{equation}
\cE^{\rm DF}_{\kappa,\alpha}(\gamma_N)\geq \tr P_{\kappa}^+(D_{\kappa}-1)P_{\kappa}^+\gamma_N+\frac\alpha2\left(1-C\alpha N^{\frac35}\right) D(\rho_{\gamma_N},\rho_{\gamma_N})
\label{eq:first_lower_bound}
\end{equation}
where $C$ depends on $\nu$. 
Noticing that $\tr P_{\kappa}^+\gamma_N P_{\kappa}^+\leq \tr \gamma_N\leq N$, we conclude that 
$$E^{\rm DF}(\kappa,\alpha,N)\geq \min_{\substack{0\leq\gamma\leq 1\\ P_\kappa^-\gamma=0\\ \tr(\gamma)\leq N}}\tr(D_\kappa-1)\gamma=\sum_{n=1}^N\left(\lambda_n(D_\kappa)-1\right).$$ 
The sum of the $N$ first eigenvalues of the Dirac-Coulomb operator on the right is explicit, since those eigenvalues are known analytically. It behaves like $N^{1/3}$. Hence we have proved, as we wanted, that 
$$E^{\rm DF}(\kappa,\alpha,N)\geq-CN^{\frac13},$$
that 
$$\tr(|D_\kappa|-1)P_\kappa^+\gamma_N P_\kappa^+\leq CN^{\frac13}$$
and that 
\begin{equation}
\alpha D(\rho_{\gamma_N},\rho_{\gamma_N})\leq CN^{\frac13}.
\label{eq:estim_final_direct_term}
\end{equation}

Inserting~\eqref{eq:estim_final_direct_term} in the error term in~\eqref{eq:estim_error_proj} we find
\begin{equation}
\bigg|\tr (P_{\kappa,\gamma_N}^+-P_\kappa^+)^2\gamma_N\bigg|+\bigg|\tr (P_{\kappa,\gamma_N}^+-P_\kappa^+)D_{\kappa}(P_{\kappa,\gamma_N}^+-P_\kappa^+)\gamma_N\bigg|\leq \frac{C}{N^{\frac1{15}}}.
\label{eq:final_estim_diff_projections}
\end{equation}
This term can therefore be neglected in the expansion of the energy up to the order $O(N^{-1/15})$. 

We now prove that 
$$\alpha D(\rho_{\gamma_N},\rho_{\gamma_N})=\alpha D(\rho_{P_\kappa^+\gamma_NP_\kappa^+},\rho_{P_\kappa^+\gamma_NP_\kappa^+})+O(N^{-1/6}).$$
To simplify our argument we introduce the densities
$$\rho_N^{\sigma,\sigma'}:=\rho_{P_\kappa^\sigma\gamma_N P_\kappa^{\sigma'}},\qquad \sigma,\sigma'\in\{\pm\}.$$
We then write
$$D(\rho_{\gamma_N},\rho_{\gamma_N})=D(\rho_N^{++},\rho_N^{++})+2D(\rho_{\gamma_N},r_N)-D(r_N,r_N)$$
with $r_N:=\rho_N^{+-}+\rho_N^{-+}+\rho_N^{--}$. We claim that 
\begin{equation}
D(r_N,r_N)\leq CN^{\frac13},
\label{eq:claim_estim_r_N}
\end{equation}
the proof of which is given below. Using that $\alpha D(\rho_{\gamma_N},\rho_{\gamma_N})=O(N^{1/3})$, we deduce from the Cauchy-Schwarz inequality for the scalar product $D(\cdot,\cdot)$ that 
$$\bigg|\alpha D(\rho_{\gamma_N},\rho_{\gamma_N})-\alpha D(\rho_N^{++},\rho_N^{++})\bigg|\leq \frac{C}{N^{\frac16}}.$$

In order to prove~\eqref{eq:claim_estim_r_N}, we show that $\|r_N\|_{L^{6/5}(\R^3)}=O(N^{1/6})$ by duality. Let $F$ be any function in $L^6(\R^3)$. Then we have
\begin{align*}
\left|\int_{\R^3}F\rho_N^{-+}\right|&=\Big|\tr \big(FP_\kappa^- \gamma_N P_\kappa^+\big)\Big|\\
&\leq \norm{F|D_0|^{-\frac12}}\norm{|D_0|^{\frac12}(P_\kappa^--P_{\kappa,\gamma}^-)|D_0|^{-\frac12}}_{\gS^2}\norm{|D_0|^{\frac12} \gamma_N}_{\gS^2}. \end{align*}
From the Hardy-Littlewood-Sobolev inequality, we have
$$\norm{F|D_0|^{-\frac12}}\leq C\norm{F}_{L^6(\R^3)}.$$
On the other hand
$$\norm{|D_0|^{\frac12} \gamma_N}_{\gS^2}^2=\tr|D_0|^{\frac12}\gamma_N^2|D_0|^{\frac12}\leq \tr|D_0|^{\frac12}\gamma_N|D_0|^{\frac12}\leq CN$$
and
$$\norm{|D_0|^{\frac12}(P_\kappa^--P_{\kappa,\gamma}^-)|D_0|^{-\frac12}}_{\gS^2}\leq C\alpha D(\rho_{\gamma_N},\rho_{\gamma_N})^{\frac12}\leq \frac{C}{N^{\frac13}}$$
by Proposition~\ref{prop:difference_projections} and~\eqref{eq:estim_final_direct_term}. This gives
$$\left|\int_{\R^3}F\rho_N^{-+}\right|\leq CN^{\frac16}\norm{F}_{L^6(\R^3)}$$
and by duality we conclude that
$$\norm{\rho_N^{-+}}_{L^{6/5}(\R^3)}\leq CN^{\frac16}.$$
The argument is the same for $\rho_N^{-+}$ and $\rho_N^{--}$, which leads to~\eqref{eq:claim_estim_r_N}, by the Hardy-Littlewood-Sobolev inequality.

As a conclusion we have shown the desired lower bound
\begin{equation}
E^{\rm DF}(\kappa,\nu,\alpha)\geq \inf_{\substack{0\leq\gamma\leq 1\\ P_\kappa^-\gamma=0\\ \alpha\tr(\gamma)\leq \nu}}\left\{\tr(D_\kappa-1)\gamma+\frac{\alpha}{2}D(\rho_\gamma,\rho_\gamma)\right\}-\frac{C}{N^{\frac1{15}}}
\end{equation}
in terms of the reduced Dirac-Fock problem projected to the positive spectral subspace of the Dirac-Coulomb operator $D_\kappa$. This concludes the proof of Proposition~\ref{prop:lower_bound}.
\end{proof}

\subsection{Upper bound}
In this section we prove the following result.

\begin{proposition}\label{prop:upper_bound}
Let $0<\kappa<1$ and $0<\nu<\nu_0$ and let $d_N$ a sequence of self-adjoint operators such that 
$$0\leq d_N \leq 1, \quad \tr\sqrt{-\Delta}\, d_N \leq C N, \quad \tr d_N \leq N$$ 
and 
$$P_{\kappa}^- d_N = 0.$$ 
Then there is a sequence $\gamma_N$ of self-adjoint operators satisfying the nonlinear constraint 
$$P_{\kappa,\gamma_N}^- \gamma_N = 0$$
such that 
$$0\leq \gamma_N \leq 1, \quad \tr\sqrt{-\Delta}\,\gamma_N \leq C N, \quad \tr \gamma_N \leq N$$
and
\begin{equation*}
\cE^{\rm DF}_{\kappa,\alpha}(\gamma_N) = \cE^{\rm DF}_{\kappa,\alpha}(d_N) + \mathcal{O}(N^{-1/15}) .
\end{equation*}
\end{proposition}

\begin{proof}
We split the proof into several steps. 

\subsubsection*{$\bullet$ S\'er\'e's retraction $\theta$}
We will use the following result of S\'er\'e \cite{Sere-09}.
\begin{theorem}\label{theo_theta}
Let $(X,\|\cdot\|_X)$ be a Banach space and $\mathcal{U}$ an open subset of $X$. Let $T : \mathcal{U} \to X$ a continuous map. We assume :
\begin{enumerate}
\item
$\mathcal{U}$ has a nonempty subset $F$ which is closed in $X$ and such that $T(F) \subset F$ ;
\item
$\exists k\in(0,1), \, \forall x \in T^{-1}(\mathcal{U}), \, \|T^2(x) -  T(x)\|_X \leq k \|T(x)-x\|_X $.
\end{enumerate}
Then there exists an open neighborhood $\mathcal{V}$ of $F$ in $X$ with $\mathrm{Fix}(T) \subset \mathcal{V} \subset \mathcal{U} , T(\mathcal{V}) \subset \mathcal{V}$ and such that for any $x \in \mathcal{V}$ , the sequence $(T^p(x))_p$ has a limit $\theta(x) \in \mathcal{V}$ for the norm $\| \cdot \|_X$, with the estimate
\begin{equation} \label{ineq_est_theta}
\forall x\in\mathcal{V}, \|\theta (x) - T^p (x)\|_X \leq \frac{k^p}{1-k} \|T(x) - x\|_X .
\end{equation}
In this way we obtain a retraction $\theta$ of $\mathcal{V}$ onto $\mathrm{Fix}(T ) \subset \mathcal{V}$ whose restriction to $F$ is a retraction of $F$ onto $F \cap \mathrm{Fix}(T )$.
\end{theorem}
Here we have denoted by $\textrm{Fix}(T)$ the fixed points of $T$. To apply S\'er\'e's result we define, for  $\kappa \in (0,1)$,
\begin{equation*}
X = \left\{ \gamma \in \mathcal{B}(\mathcal{H}), \gamma^* = \gamma, |D_\kappa|^{1/2} \gamma |D_\kappa|^{1/2} \in \s_1 \right\},
\end{equation*}
and
\begin{align*}
T : \, &X \longrightarrow \mathcal{B}(\mathcal{H}) \\
& \gamma \mapsto P_{\kappa,\gamma}^+ \gamma P_{\kappa,\gamma}^+.
\end{align*}
Let us fix some $\nu < \nu_0$ and define
\begin{equation*}
F = \left\{ 0 \leq \gamma \leq 1, \|\gamma\|_X + \lambda \|T (\gamma) - \gamma\|_X \leq M N,\quad \alpha \tr\gamma \leq \nu \right\}
\end{equation*}
for some $\lambda, M > 0$ that we will choose later. For $r>0$, we define $\mathcal{U} = F + B_X(r)$. Note that $F\neq \emptyset$ since $0\in F$.

We will first check that the assumptions of Theorem~\ref{theo_theta} are satisfied in our regime, and then we will apply it to prove Proposition~\ref{prop:upper_bound}. 

\subsubsection*{$\bullet$ Verifying the stability (Assumption 1)}

Here, we assume that the retraction property holds for some $k$. Let us check that $T(F)\subset F$. That $T\gamma \in X$ is a consequence of Hardy's inequality and Theorem~\ref{thm:Dirac_gap}. By definition of $T$ we also have directly that $0\leq T(\gamma) \leq 1$ and $\alpha \tr T(\gamma) \leq \alpha \tr \gamma \leq \nu$. It remains to verify the norm condition in the definition of $F$.  Using the triangle inequality we obtain
\begin{align*}
\|T(\gamma)\|_X + \lambda \|T^2 (\gamma) - T(\gamma)\|_X \leq \|\gamma\|_X + (1+\lambda k)\|\gamma - T(\gamma)\|_X.
\end{align*}
Choosing $\lambda > 1/(1-k)$ in the above inequality implies $T\gamma \in F$.

\subsubsection*{$\bullet$ Verifying the retraction property (Assumption 2)}
Let $\gamma \in \mathcal{U}$, we have
\begin{align*}
T^2 \gamma &= P_{\kappa,T\gamma}^+ T \gamma P_{\kappa,T\gamma}^+ \\
&= P_{\kappa,\gamma}^+ T \gamma P_{\kappa,\gamma}^+ +\left( P_{\kappa,T\gamma}^+ - P_{\kappa,\gamma}^+\right) T \gamma P_{\kappa,T\gamma}^+ + P_{\kappa,\gamma}^+ T \gamma \left( P_{\kappa,T\gamma}^+ - P_{\kappa,\gamma}^+\right).
\end{align*}
Note that $\alpha \tr T (\gamma) \leq \alpha \tr (\gamma) < \nu$ so that $|D_{\kappa,T(\gamma)}|^{1/2}$, $|D_{\kappa,\gamma}|^{1/2}$ and $|D_{\kappa}|^{1/2}$ are comparable as a consequence of Hardy's inequality (\ref{eq:upper_bd_D}) and Theorem~\ref{thm:Dirac_gap}. 

At this step we need two technical lemmas whose proofs are postponed to the end of the argument.

\begin{lemma}[H\"older inequality in weighted Schatten space]\label{lem_holder_schatten}
Let $\gamma \in X$ such that $0\leq \gamma \leq 1$. Let $0 \leq a,b \leq 1/2$ and define $q$ by $1/q = a+b$, then
\begin{equation*}
\||D_\kappa|^a \gamma |D_\kappa|^b\|_{\s^q} \leq \|\gamma\|_{X}^{a+b}
\end{equation*}
\end{lemma}

\begin{lemma}[Estimate on differences of projections in $X$]\label{lem_diff_proj_s6}
Let $ 0< \nu < \nu_0$ and $0< \kappa < 1$, then for all $\gamma_1,\gamma_2 \in \{ \gamma \in X, \|\gamma\|_X \leq M, \; \alpha \tr(\gamma) \leq \nu \}$, we have
\begin{equation*}
\||D_\kappa|^{1/2}(P_{\kappa,\gamma_1}^+-P_{\kappa,\gamma_2}^+)|D_\kappa|^{-\varepsilon}\|_{\s^6} \leq C_{\nu,\kappa,\varepsilon}  \,\alpha \| \gamma_1 - \gamma_2\|_{X}.
\end{equation*}
\end{lemma}

Now, using that $P_{\kappa,\gamma}^+ T (\gamma) P_{\kappa,\gamma}^+ = T(\gamma)$, we obtain by H\"older's inequality and \cref{lem_diff_proj_s6}, that
\begin{align*}
&\|T^2 (\gamma) - T(\gamma)\|_{X} \\
&\qquad \leq C_{\kappa,\nu} \norm{D_\kappa|^{\frac12}(P_{\kappa,T(\gamma)}^+-P_{\kappa,\gamma}^+)|D_\kappa|^{-\frac13}}_{\s^6}  \norm{|D_\kappa|^{\frac13} T(\gamma) |D_\kappa|^{\frac12}}_{\s^{6/5}}  \\
&\qquad \leq C_{\kappa,\nu} \alpha \norm{|D_\kappa|^{\frac13} T(\gamma) |D_\kappa|^{\frac12}}_{\s^{6/5}} \|T (\gamma) - \gamma\|_X,
\end{align*}
where $C_{\kappa,\nu}$ is a constant depending only on $\kappa$ and $\nu$. Using that $\gamma \in \mathcal{U}$ and \cref{lem_holder_schatten} we obtain
\begin{equation*}
 \norm{|D_\kappa|^{\frac{1}{3}} T (\gamma) |D_\kappa|^{\frac{1}{2}}}_{\s^{6/5}} \leq \|\gamma\|_X^{5/6} \leq C_{\kappa,\nu} (MN + r)^{5/6}.
\end{equation*}
Hence
\begin{equation*}
\|T^2 (\gamma) - T(\gamma)\|_{X} \leq C_{\kappa,\nu} \alpha^{1/6} \left(\nu M + \alpha r \right) \|T (\gamma) - \gamma\|_X.
\end{equation*}
This shows that for any $M,r>0$ fixed, taking $N$ sufficiently large is enough for the retraction property to hold with a retraction factor
\begin{equation}
	\label{eq:estimate_k}
	k \leq C_{\kappa,\nu} \alpha^{1/6}.
\end{equation}

\subsubsection*{$\bullet$ Conclusion of the proof of Proposition~\ref{prop:upper_bound}}

The same proof as for the lower bound, starting from (\ref{eq:decomp_kinetic}), but with $\gamma_N$ replaced by $d_N$, shows that
\begin{equation*}
\cE^{\rm DF}_{\kappa,\alpha}(T(d_N)) = \cE^{\rm DF}_{\kappa,\alpha}(d_N) + \mathcal{O}(N^{-1/15}).
\end{equation*}
Note that this does not hold for any trial state as we use intensively that $P_\kappa^- d_N = P_{\kappa,d_N}^- T(d_N) =0 $. In fact, the same argument applied $n$ times leads to 
\begin{equation*}
\cE^{\rm DF}_{\kappa,\alpha}(T^n (d_N)) = \cE^{\rm DF}_{\kappa,\alpha}(d_N) + \mathcal{O}(N^{-1/15}).
\end{equation*}
We now use the retraction $\theta$, for which we have from Theorem~\ref{theo_theta} that
\begin{equation*}
\|\theta(d_N) - T^n (d_N)\|_X \leq \frac{k^n}{1-k} \|T(d_N) -d_N\|_X.
\end{equation*}
In view of (\ref{eq:estimate_k}) it is sufficient for our purpose to show 
$${\alpha^{n/6} \|Td_N - d_N\|_X = o(N^{-1/15})}$$ 
for a certain $n$. This is clearly the case for $n = 3$ although this is not optimal since it only uses that $\|T(d_N)\|_X + \|d_N\|_X \leq CN$. We therefore obtain
\begin{align*}
\cE^{\rm DF}_{\kappa,\alpha}(\theta(d_N)) &= \cE^{\rm DF}_{\kappa,\alpha}(T^3(d_N)) + \tr(D_\kappa -1)(\theta(d_N) - T^3(d_N)) \\
&\quad
 + \frac{\alpha}{2}\left(D(\rho_{T^3(d_N)}-\rho_{\theta(d_N)},\rho_{T^3(d_N)}) + D(\rho_{T^3(d_N)}-\rho_{\theta(d_N)},\rho_{\theta(d_N)})\right) \\
 &= \cE^{\rm DF}_{\kappa,\alpha}(T^3(d_N)) + \mathcal{O}(\|\theta(d_N) - T^3(d_N)\|_X)
\end{align*}
where we used that, 
\begin{equation*}
\alpha D(\rho_{T^3(d_N)}-\rho_{\theta(d_N)},\rho_{T^3(d_N)}) \leq C \alpha \|\theta(d_N) - T^3(d_N)\|_X \|T^3(d_N)\|_X
\end{equation*}
and that $\alpha \|T^3(d_N)\|_X = \mathcal{O}(1)$. The last error term is dealt with similarly. Finally, we obtain as we wanted
\begin{align*}
\cE^{\rm DF}_{\kappa,\alpha}(\theta (d_N)) = \cE^{\rm DF}_{\kappa,\alpha}(d_N) + \mathcal{O}(N^{-1/15}),
\end{align*}
which concludes the proof of Proposition~\ref{prop:upper_bound}.
\end{proof}

It remains to provide the 
\begin{proof}[Proof of Lemma~\ref{lem_holder_schatten}]
Define $p_a = 1/a$ and $p_b = 1/b$. By H\"older inequality we have
\begin{equation*}
\||D_\kappa|^a \gamma |D_\kappa|^b\|_{\s_q} \leq \||D_\kappa|^a \gamma^{1/2}\|_{\s_{p_a}} \| \gamma^{1/2}|D_\kappa|^b\|_{\s_{p_b}}.
\end{equation*}
We bound each of the factors above using the Araki-Lieb-Thirring inequality together with the fact that $0\leq \gamma \leq 1$. We have
\begin{align*}
\| |D_\kappa|^a \gamma^{1/2}\|_{\s_{p_a}} &= \left(\tr\left(|D_\kappa|^a \gamma |D_\kappa|^a\right)^{p_a/2}\right)^{1/p_a} \\
&\leq \left(\tr|D_\kappa|^{1/2} \gamma^{p_a/2} |D_\kappa|^{1/2} \right)^{1/p_a} \\
&\leq \left(\tr|D_\kappa|^{1/2} \gamma |D_\kappa|^{1/2} \right)^{1/p_a} \\
&\leq \|\gamma\|_{X}^a.
\end{align*}
The same proof holds for the other term and gives the desired result.
\end{proof}

\begin{proof}[Proof of Lemma~\ref{lem_diff_proj_s6}] 
We use Stone's formula and the resolvent identity to express the difference as
\begin{multline*}
|D_\kappa|^{1/2} \left(P_{\kappa,\gamma_1} - P_{\kappa,\gamma_2}\right) |D_\kappa|^{-\varepsilon} \\ =\frac1{2\pi}\int_{\mathbb{R}} |D_\kappa|^{1/2} \frac{1}{D_{\kappa,\gamma_1} + i\eta} \alpha \left((\rho_{\gamma_2} - \rho_{\gamma_1}) \ast \frac{1}{|x|} \right)\frac{1}{D_{\kappa,\gamma_2} + i\eta}|D_\kappa|^{-\varepsilon}d\eta.
\end{multline*}
The argument is now exactly the same as for~\eqref{eq:estim_proj_S6}.
\end{proof}

\subsection{Conclusion of the proof of Theorem~\ref{thm:Scott}} 

Up to this point we have not used that $\kappa=\nu$. A corollary of Handrek and Siedentop's article \cite{HanSie-15} is that the projected Dirac-Fock problem behaves as
\begin{multline}
\inf_{\substack{0\leq\gamma\leq 1\\ P_\kappa^-\gamma=0\\ \tr(\gamma)\leq N}}\left\{\tr(D_\kappa-1)\gamma+\frac{\alpha}{2}D(\rho_\gamma,\rho_\gamma)\right\}\\
 =e_{\rm TF}(1)\, \alpha^2 N^{7/3} + c_{\textrm{Scott}}(\kappa) + \mathcal{O}(\alpha^2 N^{47/24}),
 \label{eq:behavior_projected}
\end{multline}
for $\kappa=\nu$. We quickly describe the argument but refer to~\cite{HanSie-15} for details.

Let $\rho_{N}^{\textrm{TF}}$ be the minimizer of the Thomas-Fermi functional and define
\begin{equation*}
\chi(x) = \alpha^4 \int_{|x-y|>\alpha^{-1}R_Z(\alpha x)} \frac{\rho_{N}^{\textrm{TF}}(\alpha y)}{|x-y|} \, dy
\end{equation*}
where $R_Z(x)$ is such that
\begin{equation*}
\int_{|x-y|>R_Z(x)} \rho_{N}^{\textrm{TF}}(y) \dd y = \frac{1}{2}.
\end{equation*}
Using 
\begin{equation*}
\frac{1}{2}D(\rho_{\gamma},\rho_\gamma) \geq D(\rho_{N}^{\textrm{TF}},\rho_\gamma) - \frac{1}{2}D(\rho_{N}^{\textrm{TF}},\rho_{N}^{\textrm{TF}})
\end{equation*}
and that $\chi \leq \rho_{N}^{\textrm{TF}}\ast \frac{1}{|x|}$, one finds that 
\begin{align*}
\tr(D_\kappa-1)\gamma+\frac{\alpha}{2}D(\rho_\gamma,\rho_\gamma)&\geq \tr (D_\kappa + \chi -1)\gamma - \frac{1}{2}D(\rho_{N}^{\textrm{TF}},\rho_{N}^{\textrm{TF}})\\
&\geq \sum_{j=1}^N\lambda_j\left(P_\kappa^+(D_\kappa + \chi -1)P_\kappa^+\right)- \frac{1}{2}D(\rho_{N}^{\textrm{TF}},\rho_{N}^{\textrm{TF}}).
\end{align*}
It is proved in~\cite[Sec.~3]{HanSie-15} that the right hand side behaves as 
$$e_{\rm TF}(1)\, \alpha^2 N^{7/3} + c_{\textrm{Scott}}(\kappa) + \mathcal{O}(\alpha^2 N^{47/24})$$ 
and this gives the lower bound in~\eqref{eq:behavior_projected}.

To obtain the upper bound, the authors of~\cite{HanSie-15} construct a state $d_N$ which satisfies the assumptions of Proposition~\ref{prop:upper_bound} and is such that
\begin{equation}\label{ineq_trial_state}
\tr(D_\kappa -1) d_N + \frac{\alpha}2 D(\rho_{d_N},\rho_{d_N})  = e_{\rm TF}(1)\, \alpha^2 N^{7/3} + c_{\textrm{Scott}}(\kappa) + \mathcal{O}(\alpha^2 N^{47/24}).
\end{equation}
Hence the final result follows, for $\kappa=\nu$ from our lower bound in Proposition~\ref{prop:lower_bound} and from the construction of the trial state $\gamma_N$ from $d_N$ in Proposition~\ref{prop:upper_bound}. This concludes the proof of Theorem~\ref{thm:Scott}.\qed

\begin{remark}
Should the limit be proven for the Dirac-Fock projected energy with $P_\kappa^+$ for $\nu\neq\kappa$, our result would immediately apply to the unprojected Dirac-Fock theory, with the same value of $\kappa$ and $\nu$. 
\end{remark}

% \bibliographystyle{my-alpha}
% \bibliography{biblio}

\newcommand{\etalchar}[1]{$^{#1}$}

\end{document}